\documentclass[preprint,12pt]{elsarticle}

\usepackage[displaymath,mathlines]{lineno}
 
\usepackage{amsmath}
\usepackage{amssymb}
\usepackage{amsbsy}
\usepackage{amsthm}

\newcommand*\patchAmsMathEnvironmentForLineno[1]{%
  \expandafter\let\csname old#1\expandafter\endcsname\csname #1\endcsname
  \expandafter\let\csname oldend#1\expandafter\endcsname\csname end#1\endcsname
  \renewenvironment{#1}%
     {\linenomath\csname old#1\endcsname}%
     {\csname oldend#1\endcsname\endlinenomath}}%
\newcommand*\patchBothAmsMathEnvironmentsForLineno[1]{%
  \patchAmsMathEnvironmentForLineno{#1}%
  \patchAmsMathEnvironmentForLineno{#1*}}%
\AtBeginDocument{%
\patchBothAmsMathEnvironmentsForLineno{equation}%
\patchBothAmsMathEnvironmentsForLineno{align}%
\patchBothAmsMathEnvironmentsForLineno{flalign}%
\patchBothAmsMathEnvironmentsForLineno{alignat}%
\patchBothAmsMathEnvironmentsForLineno{gather}%
\patchBothAmsMathEnvironmentsForLineno{multline}%
}

\usepackage{tabularx}
\newcommand{\Real}{\mathbb{R}}
\newcommand{\T}{^\mathrm{T}}
\newtheorem{theorem}{Theorem}
\newtheorem{proposition}[theorem]{Proposition}
\renewcommand{\max}{\mathrm{max}} 
\renewcommand{\min}{\mathrm{min}} 
\newcommand{\ext}{\mathrm{ext}} 
\newcommand{\cat}{\mathrm{cat}} 
\newcommand{\bio}{\sf{ }} 
\newcommand{\dt}{\mathrm{d}t}

\newcommand{\lc}{\left[}

\newcommand{\rc}{\right]}
\newcommand{\lp}{\left(}
\newcommand{\rp}{\right)}

\usepackage{footmisc}
\DefineFNsymbols*{sascha}{{\!\!1}{\!\!2}{\!\!3}{\!\!4}}
\setfnsymbol{sascha}

\journal{Elsevier}

\begin{document}

\begin{frontmatter}



\title{Dynamic optimization of metabolic networks coupled with gene expression}


\author[label1]{Steffen~Waldherr}
\address[label1]{Institute for Automation Engineering, 
	Otto von Guericke University Magdeburg, 
	Universit\"atsplatz 2, 39106~Magdeburg, Germany, steffen.waldherr@ovgu.de}
\author[label2]{Diego~A.~Oyarz\'un}
\address[label2]{Department of Mathematics, 
	Imperial College London, 
	SW7 2AZ London, United Kingdom}
\author[label3]{Alexander~Bockmayr}
\address[label3]{DFG Research Center \textsc{Matheon}, 
	Freie Universit\"at Berlin,
	Arnimallee 6, 14195~Berlin, Germany}
\begin{abstract}
The regulation of metabolic activity by tuning enzyme expression levels is crucial to sustain cellular growth in changing environments.
Metabolic networks are  often studied at steady state using constraint-based models and optimization techniques.
However, metabolic adaptations driven by changes in gene expression cannot be analyzed by steady state models, as these do not account for temporal changes in biomass composition.

Here we present a dynamic optimization framework that integrates the metabolic network with the dynamics of biomass production and composition.
An approximation by a timescale separation leads to a coupled model of quasi steady state constraints on the metabolic reactions, and differential equations for the substrate concentrations and biomass composition.
We propose a dynamic optimization approach to determine reaction fluxes for this model, 
explicitly taking into account enzyme production costs and enzymatic capacity.
In contrast to the established dynamic flux balance analysis, our approach allows predicting dynamic changes in both the metabolic fluxes and the biomass composition during metabolic adaptations.
Discretization of the optimization problems leads to a linear program that can be efficiently solved.

We applied our algorithm in two case studies: a minimal nutrient uptake network, and an abstraction of core metabolic processes in bacteria.
In the minimal model, we show that the optimized uptake rates reproduce the empirical Monod growth for bacterial cultures.
For the network of core metabolic processes, the dynamic optimization algorithm predicted commonly observed metabolic adaptations, such as a diauxic switch with a preference ranking for different nutrients, re-utilization of waste products after depletion of the original substrate, and metabolic adaptation to an impending nutrient depletion.
These examples illustrate how dynamic adaptations of enzyme expression can be predicted solely from an optimization principle.

\end{abstract}


\end{frontmatter}

\section{Introduction}
\label{sec:introduction}

A key aspect of cellular dynamics is the ability to adapt metabolic activity to changing environments.
This involves a dynamic re-organization of enzyme expression levels, in order to accommodate for variability in nutrient abundance and environmental shocks that have a deleterious impact on growth. 
These adaptations emerge from a complex array of regulatory interactions between metabolism and the genetic machinery. 
Since many of these interactions are unknown or incompletely understood, a fully mechanistic grasp of how they control metabolic adaptations is currently beyond our reach. 
Moreover, the analysis of large-scale mechanistic models is typically hampered by the high number of molecular species and parameters involved.

An alternative approach to predict metabolic adaptations is to assume an underlying optimality principle \cite{Watson1986,FellSma1986,VarmaPal1994}. 
Numerous studies have considered metabolic adaptations in microbes by computing optimal metabolic fluxes in a stoichiometric model under a suitable objective function \cite{RielGiu2000,CovertPal2002,MeadowsKar2010,SteuerKno2012}.
Stoichiometric models are a structural description of a metabolic network and cannot provide information on the enzyme concentrations.
Several approaches have attempted to overcome this by integrating gene regulation with stoichiometric models, either by modelling enzyme expression qualitatively with Boolean variables describing regulatory effects \cite{CovertSch2001,CovertPal2002}, or by explicitly including enzyme capacity constraints in the optimization problems \cite{GoelzerFro2011}.
An alternative approach are the cybernetic models \cite{RamkrishnaSon2012}, where regulation is explicitly modelled and assumed to optimize a cellular objective.
While the classical cybernetic approach explicitly includes reaction kinetics in the model, a hybrid approach has been suggested \cite{KimVar2008}, where the rates are determined from flux balance analysis.

Models that integrate metabolism and gene expression can potentially yield better predictions than those focused on metabolism in isolation.
This can be particularly helpful in metabolic adaptations caused by environmental fluctuations.
To capture the dynamics of biomass and gene expression linked to metabolic activity,
previous studies have mostly used ad-hoc combinations of various modeling frameworks.
Examples are combinations of constraint-based steady state models with ordinary differential equations or Boolean regulatory logic \cite{VarmaPal1994,CovertPal2002,CovertXia2008}.
This combination approach has been successful in proposing integrated models up to whole-cell dynamics \cite{KarrSan2012}.

In this paper, we propose a dynamic modeling framework for metabolic networks coupled with gene expression of enzymes and production of other macromolecules.
We develop an optimization algorithm to predict optimal time courses for nutrient uptake, metabolic fluxes, and gene expression rates in such networks.

The classical approach to constraint-based optimization of metabolic fluxes, commonly called Flux Balance Analysis (FBA), relies on an optimization problem with algebraic constraints stemming from a steady state restriction \cite{VarmaPal1994a,ReedPal2003,OrthThi2010}. 
Mathematically, the FBA approach in the simplest form leads to a linear program of the form
\begin{equation}
  \label{eq:fba-lp}
  \max_v \{c\T v \mid  S v = 0, \enspace v_{\min} \leq v \leq v_{\max}\},  
\end{equation}
where $v$ is the reaction flux vector, $c$ a biomass weighting vector, $S$ the stoichiometric matrix, and $v_{\min}$, $v_{\max}$ are lower and upper component-wise bounds on the fluxes, respectively.
While the most common optimization objective is the maximization of biomass production, an experimental evaluation also highlighted additional biologically relevant objectives \cite{SchuetzKue2007}.

One point of critique to FBA is its coarse description of the biomass composition.
While growth-dependent changes in the biomass composition have been taken into account in the past \cite{PramanikKea1997}, constraints related to the actual biomass composition by enzymes or other cellular macromolecules are usually not considered. 
At least on the level of individual metabolic pathways, there is good evidence that the enzyme production cost is an important factor in the regulation of these pathways \cite{WesselyBar2011}.
Thus, it seems plausible that the inclusion of biomass composition and enzyme costs in metabolic optimization can potentially improve the quality of its predictions.
As an extension to FBA in this direction, the resource balance analysis (RBA) approach has been proposed \cite{GoelzerFro2011}. 
This includes the conversion of metabolites into specific enzymes and other proteins in the network, and adds the enzymatic capacity as constraint on metabolic fluxes for the optimization.
RBA yields a linear optimization problem and can intrinsically describe changes in both the growth rate and biomass composition due to environmental changes from an optimization principle alone.
A conceptually equivalent approach has been proposed independently in \cite{LermanHyd2012} under the term ME (metabolism and macromolecular expression) model.
Both approaches are however limited to situations of steady exponential growth.

For batch processes or in changing environments, the model needs to go beyond the stationary approach and account for dynamic changes in metabolic activity.
FBA has been used to predict dynamic changes in biomass and nutrients using iterative approaches \cite{VarmaPal1994}.
However, the iterative optimization uses a steady state constraint and does not account for the model dynamics, and thus the predictions may not be optimal in changing environments.
Dynamic effects are physiologically important, as is evidenced by the experimental observation that even in steady state, cells would show flux distributions which are slightly suboptimal, but which allow for easier transitions to other environmental conditions \cite{SchuetzZam2012}.
Also, the numerical accuracy of the iterative approach can at best be evaluated heuristically or by numerical experimentation, unless specialized numerical algorithms are applied \cite{HoffnerHar2013}.

By formulating an appropriate dynamic optimization problem, it is possible to compute optimal fluxes over the whole time range of interest.
This approach has been proposed in dynamic flux balance analysis (dFBA) \cite{MahadevanEdw2002}. 
In dFBA, one can distinguish between a ``static optimization approach (SOA)'', similar to the previously used iterative FBA \cite{VarmaPal1994}, and a ``dynamic optimization approach (DOA)''.
The static approach is useful to get feasible nutrient and biomass dynamics under metabolic constraints, but it cannot resolve the optimization problem over the complete timescale of interest.
The dynamic approach DOA directly considers an objective function which depends on the dynamics over the complete timescale, potentially under dynamic metabolic constraints, and thus provides a consistent solution to the dynamic optimization problem.
However, in the same way as classical FBA, dynamic FBA uses only a coarse description of biomass composition.
Biomass is captured only as one component, and different allocations of biomass to different metabolic tasks, such as considered in RBA, cannot be represented. 

In the study described here, we developed a mathematical framework for dynamic models of coupled metabolism and gene expression.
We denote such models with the term \emph{metabolic-genetic networks}.
From a rigorous timescale separation, we approximate this model by a quasi steady state, constraint-based part for the intracellular metabolism, and a dynamic part for the evolution of biomass and substrate concentrations.
For this model class, we developed a dynamic optimization approach, called \emph{dynamic enzyme-cost FBA} (deFBA), that includes a detailed description of biomass and accounts for the enzyme cost.
The deFBA method respects biophysical constraints motivated from resource balance analysis \cite{GoelzerFro2011}, and additionally includes dynamic changes in biomass composition and substrate concentrations.

In Table~\ref{tab:fba-scheme} we compare our deFBA approach to the established methods available in the literature.
The distinction among the methods is based on two criteria: one for the type of optimization approach (static, iterative, or dynamic) and one for whether the enzyme production cost is taken into account for the optimization or not.
The comparison highlights the improved generality of deFBA compared to established methods with respect to these two criteria.

\begin{table}[tbp]
\caption{
\bf{Different flux optimization approaches}}
\begin{footnotesize}
  \begin{tabular}{l||c|c|c|}
    \bf{Optimization approach} & No enzyme cost & Enzyme cost included \\ \hline
    \bf{static} & FBA \cite{VarmaPal1994a} & RBA \cite{GoelzerFro2011}, ME networks \cite{LermanHyd2012} \\
    \bf{iterative} & iFBA \cite{VarmaPal1994}, dFBA (SOA) \cite{MahadevanEdw2002} & \\
    \bf{dynamic}  & dFBA (DOA) \cite{MahadevanEdw2002} & deFBA \emph{(this paper)}
  \end{tabular}
\end{footnotesize}
\begin{flushleft}
\end{flushleft}
\label{tab:fba-scheme}
\end{table}

We applied the deFBA method to two exemplary metabolic-genetic networks, a minimal nutrient uptake network, and a larger network modeled that describes core cellular processes in bacteria.
For the example of a minimal metabolic-genetic network, we evaluated the dynamics resulting from an optimization with a set of biologically meaningful objective functionals.
Using a Michaelis-Menten reaction rate for the substrate uptake, we showed that the minimal metabolic-genetic network is equivalent to the empirical Monod growth kinetics. 
We also observed a close similarity between optimal solutions and Monod kinetics when minimizing the time for substrate metabolization or maximizing the biomass integral with deFBA.
We argue that this observation supports the biological validity of these objective functionals.

For the larger network of core cellular processes, we focussed on the discounted biomass integral as objective function.
We analyzed different scenarios of nutrient availability which are relevant in a biotechnological setting, including the switch from one carbon source to another and growth under oxygen limitation.
With the dynamic optimization approach, we observed clearly distinguished growth phases, obtained biologically reasonable adaptation dynamics upon changes in nutrient availability, and could predict different dynamic biomass compositions of the cells depending on the growth conditions.

In summary, the proposed deFBA method for metabolic networks coupled with gene expression allows us to infer dynamic adaptations of the cellular metabolic state from biophysical capacity constraints under an optimality principle.
The approach can predict metabolic changes ocurring in cellular adaptations to a dynamic environment, without knowledge of the involved regulatory mechanisms.

\section{Modeling and optimization of metabolic-genetic networks}
\label{sec:results}

\subsection{Model construction}
\label{sec:model-metab-genet}

Our dynamic optimization algorithm is based on a dynamic mass balance model of a metabolic-genetic network. 
We modeled metabolic-genetic networks with three types of molecular species:
\begin{itemize}
\item Extracellular nutrients and waste, with the molar amount vector $Y$;
\item Intracellular metabolites, with the molar amount vector $X$;
\item Macromolecules like gene products or large metabolites forming cellular building blocks, with the molar amount vector $P$.
\end{itemize}
We split the network reactions accordingly into three classes:
\begin{itemize}
\item Exchange reactions, with fluxes $V_y$, between the cell and the environment;
\item Metabolic reactions, with fluxes $V_x$, converting one set of metabolites into another one;
\item Biomass reactions, with fluxes $V_p$, converting metabolites into macromolecules or vice versa, for example gene expression or anabolic reactions.
\end{itemize}

We denote the vector of all reaction fluxes as
\begin{equation}
  \label{eq:flux-vector}
  V = (V_y, V_x, V_p)\T.
\end{equation}

We assumed two general properties of such networks, lead to the time scale separation employed in this paper (described in the Section~\ref{sec:time-scale-appr}):
\begin{itemize}
\item Each macromolecule is composed of a large number of small metabolites.
For example, a simple production reaction for a macromolecule $P$ from a single metabolite $X$ may be represented as $\alpha X \rightarrow P$, where $\alpha$ is a large stoichiometric coefficient.
\item The biomass reactions, i.e., the production of macromolecules, get proportionally slower as the relative stoichiometry $\alpha$ of macromolecules to metabolites increases.
This property is accounted for by scaling the biomass reaction fluxes $V_p$ with a small dimensionless factor $\varepsilon$ in the dynamic model.
\end{itemize}

For the purpose of timescale separation, we assumed that each reaction flux can be expressed as a time-varying function $V_i(t,y,x,P)$ of the species concentrations, 
where $y$ and $x$ are the concentrations of extracellular species and intracellular metabolites, respectively.
The explicit time-dependence of $V_i$ represents modulation of enzyme activity and gene expression due to cellular signalling.
Assuming that the extracellular volume $\vartheta_e$ is constant in time and the cellular volume $\vartheta_c(t,P)$ is a time-varying function of the amount of macromolecules $P$, the concentrations $x$ and $y$ are defined as
\begin{equation}
  \label{eq:concentrations}
  y = \frac{Y}{\vartheta_e} \qquad 
  x = \frac{X}{\vartheta_c(t,P)}.
\end{equation}
The reaction fluxes are considered to be given in units of molar amount per time, as is generally recommended for models with multiple compartments, for example in the SBML specification \cite{HuckaFin2003}.
This is reflected by the assumption that the reaction flux $V_i(t,y,x,P)$ depends on the molar amount of the enzymes, and not on their concentration.

For a general network of this type, we derived the differential equations for the network dynamics from mass balancing as
\begin{equation}
  \label{eq:metabolic-genetic-network}
  \begin{aligned}
    \dot{Y} &= - S^y_y \,V_y \\
    \dot{X} &= S^x_y \,V_y + S^x_x \,V_x - \alpha \;\varepsilon\, S^x_p\, V_p \\
    \dot{P} &= \varepsilon\, S^p_p\, V_p,
  \end{aligned}
\end{equation}
where the matrices $S^i_j$, $i,j \in \{ x, y, p \}$ describe the stoichiometry of species $i$ in reactions $V_j$.

\subsection{Timescale approximations of metabolic-genetic networks}
\label{sec:time-scale-appr}

\subsubsection{Transformation to the singular perturbation normal form}

Based on the previous transformation~\eqref{eq:concentrations} to units of concentration for the extracellular species $y$, we rewrote the metabolic-genetic network model~\eqref{eq:metabolic-genetic-network} as
\begin{equation}
  \label{eq:metabo-gen-net-intensive}
  \begin{aligned}
    \dot{y}(t) &= - \frac{1}{\vartheta_e}\, S^y_y\, V_y(t,y,x,P) \\
    \dot{X}(t) &= S^x_y\, V_y(t,y,x,P) + S^x_x\, V_x(t,x,P) - \alpha\, \varepsilon\, S^x_p\, V_p(t,x,P) \\
    \dot{P}(t) &= \varepsilon\, S^p_p\, V_p(t,x,P),
  \end{aligned}
\end{equation}
where $x = X/\vartheta_c(t,P)$.
Importantly, we did not transform the intracellular metabolic state $X(t)$ to concentrations in the left hand side of~\eqref{eq:metabo-gen-net-intensive}, in order to avoid the non-linearity in the differential equations that would result from a time-varying volume during growth.

Based on the assumptions in Section~\ref{sec:model-metab-genet}, the time scale separation is expressed mathematically as the limit $\alpha \rightarrow \infty$, $\varepsilon \rightarrow 0$, with the product $\alpha \varepsilon$ staying constant.
Since $\alpha$ describes the ratio of biomass molarity to the consumed nutrient molarity, when $\alpha$ tends to infinity with a finite initial nutrient supply $Y(0)$, the model would have a trivial solution where no biomass can be produced.
This problem can be avoided by assuming that the extracellular volume $\vartheta_e \rightarrow \infty$, with both the product $\varepsilon \vartheta_e$ and the extracellular nutrient concentrations $y$ remaining finite, i.e., the total nutrient amount is in proportion with the achievable biomass.
In summary, we assume 1) a large ratio for the specific molecular mass of macromolecules to metabolites, 2) slow reactions that produce macromolecules, and 3) a large extracellular volume compared to the cellular volume.

A long time scale is then defined by
\begin{equation}
  \label{eq:long-timescale}
  T = \varepsilon\, t.
\end{equation}
Assuming that the reaction rates $V_i(t,y,x,P)$ and the cellular volume $\vartheta_c(t,P)$ are slowly varying with respect to their explicit dependence on $t$, we transformed the model~\eqref{eq:metabo-gen-net-intensive} to the long time scale (see equation~\eqref{eq:net-long-timescale} in~\ref{sec:derivation-long-time-scale}).
This assumption will generally be valid for regulation of enzymatic activity by slowly changing environmental conditions, gene expression, and changes in cellular morphology, as these are expected to act on a similar time scale as the accumulation of macromolecules $P$.
Note that the framework can nevertheless accomodate for fast allosteric regulation, since that would typically be modelled by time-invariant kinetics depending on $x$ only.
As an example, a typical rate model for an enzymatic reaction with enzyme $P$, substrate $x_1$, and allosteric inhibitor $x_2$ is given by $V = V_{max}(P) x_1/((K_1 + x_1) (K_2 + x_2))$ \cite{Fall2002}, which does not depend explicitly on time and thus is trivially ``slowly varying'' in the sense required here.

With a slight abuse of notation, we will represent $V_i$ and $\vartheta_c$ being slowly varying in the following by writing $V_i(T,y,x,P)$ and $\vartheta_c(T,P)$.

\subsubsection{Derivation of a quasi steady state model}
\label{sec:rigor-deriv-quasi}

The model on the long time scale is in the standard form to perform singular perturbation by Tikhonov's theorem \cite{Khalil2002}.
To apply this theorem, the following conditions need to be verified:
\begin{description}
\item[Condition 1] The quasi steady state equation in the limit $\varepsilon \rightarrow 0$, given by setting the right hand side of $\dot X$ to zero in~\eqref{eq:metabo-gen-net-intensive},
  needs to have a locally unique solution 
  \begin{equation}
    \label{eq:quasi-steady-state}
    X = q(T,y,P)
  \end{equation}
  for all admissible values of $T$, $y$, and $P$.
\item[Condition 2] The quasi steady state~\eqref{eq:quasi-steady-state} needs to be an exponentially stable steady state of the boundary layer model (i.e., the fast dynamics of~\eqref{eq:metabo-gen-net-intensive}), which is given by
  \begin{equation}
    \label{eq:boundary-layer}
    \dot X = S^x_y\, V_y(T,y,x,P) + S^x_x\, V_x(T,x,P) - \alpha\, \varepsilon\, S^x_p\, V_p(T,x,P),
  \end{equation}
  uniformly in $T$, $y$, and $P$.
  Note that for stability analysis, $T$, $y$ and $P$ are considered as constants in~\eqref{eq:boundary-layer}.
\end{description}

For the technical details behind these conditions, we refer to \cite[chapter 11]{Khalil2002}.

If these conditions hold, one can approximate solutions of the original model by a reduced model, where the fast dynamics $\dot X$ are considered to be in quasi steady state, and the variable $X$ is replaced by its quasi steady state solution~\eqref{eq:quasi-steady-state} in the dynamics of the slow variables $y$ and $P$.
For a rigorous presentation of the reduced model, see equation~\eqref{eq:reduced-model-long-timescale} in~\ref{sec:derivation-long-time-scale}.
Solutions of the original model on the long time scale are denoted by $(y(T,\varepsilon),\ X(T,\varepsilon),\ P(T,\varepsilon))$, and solutions of the reduced model by $(y(T,0),\ P(T,0))$.
Application of Tikhonov's theorem \cite[Theorem 11.1]{Khalil2002} then yields the following result.

\begin{theorem}
  \label{theorem:tikhonov}
  If the metabolic-genetic network model satisfies Conditions 1 and 2 above and the reduced model~\eqref{eq:reduced-model-long-timescale} has a unique solution $(y(T,0),\ P(T,0))$,
  then there exists $\varepsilon^\ast > 0$ such that for all $\varepsilon < \varepsilon^\ast$,
  the original model~\eqref{eq:net-long-timescale} has a unique solution $(y(T,\varepsilon),\ X(T,\varepsilon),\ P(T,\varepsilon))$ on a finite time interval $0 \leq T \leq T^\ast(\varepsilon)$, and this solution satisfies the bound
  \begin{equation}
    \label{eq:qssa-solution-bound}
    \| y(T,\varepsilon) - y(T,0) \| + \| P(T,\varepsilon) - P(T,0) \| = \mathcal{O}(\varepsilon).
  \end{equation}
  In addition, for any $T^{\ast\ast} > 0$, there exists $\varepsilon^{\ast\ast}$ with $0 < \varepsilon^{\ast\ast} < \varepsilon^\ast$ such that the solution in the fast variable satisfies
  \begin{equation}
    \label{eq:fast-variable-bound}
    \| X(T,\varepsilon) - q(T,y(T,0),P(T,0)) \| = \mathcal{O}(\varepsilon)
  \end{equation}
  for all $T \in [ T^{\ast\ast}, T^\ast(\varepsilon) ]$ and $\varepsilon < \varepsilon^{\ast\ast}$.
\end{theorem}

Theorem~\ref{theorem:tikhonov} says that solutions of the reduced model~\eqref{eq:reduced-model-long-timescale} can be used to approximate the slow variables of the original model~\eqref{eq:net-long-timescale} for small values of $\varepsilon$,
and the quasi steady state solution $q(T,y,P)$ can be used to approximate the fast variable after a short transient.
The approximation requires that the boundary layer model, i.e., the dynamics of the metabolic network, has a unique, exponentially stable steady state.
Note that the approximation is generally valid on a finite time scale only, but this does not pose a problem here, since we are considering optimization over a finite time span.

Regarding Condition 1, explicit conditions for the existence of a unique steady state have been presented in \cite{OyarzunSta2012,GoelzerFro2014} for the special case of unbranched metabolic pathways with specific feedback regulatory mechanisms.
Regarding Condition 2, the stability analysis of metabolic networks is a field of active research, with most stability conditions being only applicable to networks with simple stoichiometries \cite{MeslemFro2011,OyarzunSta2012}.

In order to show rigorously that a given metabolic-genetic network can be well approximated by the quasi steady state model derived here, one would have to check Conditions 1 and 2 based on a kinetic model.
These conditions, however, are hard to check in realistic networks because enzyme kinetics are not always known and because it may not be possible to solve for the steady state of $X$.
As is commonly done in constraint-based models, we typically need to assume existence and stability of a quasi steady state based on biophysical insight.
However, there are cases where unstable dynamics have been shown for metabolic processes, for example oscillations in glycolysis \cite{MadsenDan2005}.

\subsubsection{A minimal metabolic-genetic network for nutrient uptake}
\label{sec:minim-metab-genet}

As a minimal example for a metabolic-genetic network, we considered a nutrient uptake network composed of one nutrient $Y$, one intracellular metabolite $X$, and one gene product $P$.
The minimal network consists of an uptake reaction $V_y$ and a biomass reaction $V_p$ as follows:
\begin{equation}
  \label{eq:minimal-network}
  \begin{aligned}
    V_y: &\ \;\;\;Y&\rightarrow X: &\qquad V_y(t,y,P) = P f_y(y) \\
    V_p: &\ \alpha X&\rightarrow P: &\qquad V_p(t,x,P) = P f_p(x),
  \end{aligned}
\end{equation}
where each reaction rate $V_{y}$, $V_p$ is split into the amount of gene product $P$ acting as enzyme and a kinetic term $f_{y}$, $f_p$, respectively.
With the classical Michaelis-Menten model, one would for example have $f_{y}(y) = k_{\cat} y / (K_y + y)$.

The quasi steady state constraint on the long timescale obtained from setting dynamics of $X$ to zero is given by
\begin{equation}
  \label{eq:quasi-steady-state-minimal}
  P f_y(y) = \alpha \varepsilon P f_p\left(\frac{X}{\vartheta_c(T,P)}\right)
\end{equation}
for this network.
If $f_p$ is invertible in the relevant domain and $P \neq 0$, a quasi steady state solution~\eqref{eq:quasi-steady-state} for $X$ is computed as
\begin{equation}
  \label{eq:quasi-steady-state-solution-minimal}
  q(T,y,P) = f_p^{-1}\left( \frac{f_y(y)}{\alpha \varepsilon} \right) \vartheta_c(T,P).
\end{equation}
The boundary layer model~\eqref{eq:boundary-layer} is 
\begin{equation}
    \label{eq:boundary-layer-minimal}
    \dot X = P\, f_y(y) - \alpha\, \varepsilon\, P\, f_p\left(\frac{X}{\vartheta_c(T,P)}\right).
\end{equation}
Since the boundary layer model describes the short time scale, $P$, $y$, and $T$ can be considered as constants.
A stability analysis of the linear approximation shows that the steady state is locally exponentially stable uniformly in $T$, $y$, and $P$ if there exists a constant $\delta > 0$ independent of $T$, $y$, and $P$, such that
\begin{equation}
  \label{eq:minimal-model-stability-condition}
  \frac{\alpha\, \varepsilon\, P}{\vartheta_c(T,P)} f_p^\prime\left( \frac{q(T,y,P)}{\vartheta_c(T,P)} \right) > \delta.
\end{equation}
This condition will be satisfied if $P / \vartheta_c(T,P)$, i.e., the gene product concentration, is bound away from zero and the function $f_p(x)$ is strictly increasing in $x$.
Assuming a non-vanishing gene product concentration, the stability condition is satisfied for any kinetics $f_p(x)$ that is strictly increasing in the substracte concentration $x$.

Finally, the approximate model for the minimal network~\eqref{eq:minimal-network} on a long time scale is given by
\begin{equation}
  \label{eq:reduced-minimal-model-fluxes}
  \begin{aligned}
    \dot{Y} &= - V_y \\
    \dot{P} &= \varepsilon V_p,
  \end{aligned}
\end{equation}
together with the quasi steady state constraint 
\begin{equation}
  \label{eq:minimal-qss-equation}
  V_y - \alpha\, \varepsilon\,V_p = 0.
\end{equation}
Note that we have gone back to units of molar amounts for the extracellular metabolites $Y$ and the original time scale $t$ for easier biological interpretation.

\subsection{Dynamic optimization in metabolic-genetic networks}
\label{sec:dynam-optim-metab}

We developed a new dynamic optimization approach to predict the time-courses of fluxes, substrate concentrations, and biomass in metabolic-genetic networks, which we call \emph{dynamic enzyme-cost FBA} (deFBA).
We included constraints on enzyme capacity and biomass composition as in resource balance analysis (RBA) \cite{GoelzerFro2011} or ME networks \cite{LermanHyd2012}, but considered a dynamic flux optimization problem as in the dynamic approach of dFBA \cite{MahadevanEdw2002}.

While in principle it is possible to apply a constraint-based dynamic flux optimization directly on the full metabolic-genetic network model~\eqref{eq:metabolic-genetic-network},
the time scale separation presented in Section~\ref{sec:time-scale-appr} can greatly help in reducing the complexity of the optimization problem and thus increase the efficiency of the solvers.
First, the time scale separation reduces the dimensionality of the problem by removing the metabolic variables $X$; second, optimizing only on the slow dynamics allows for a coarser time discretization grid.

In the reduced model~\eqref{eq:reduced-model}, we still face the problem that we typically do not know the exact kinetics for the reaction fluxes $V(T,y,x,P)$.
We can circumvent this problem by considering the fluxes $V_i$ as free variables in the optimization, as is commonly done in constraint based models.

As shown in Section~\ref{sec:rigor-deriv-quasi} and~\ref{sec:derivation-long-time-scale}, the reduced model to be used in the optimization is constructed as
\begin{equation}
  \label{eq:reduced-model-fluxes}
  \begin{aligned}
    \dot{Y} &= - S^y_y \,V_y \\
    \dot{P} &= \varepsilon \,S^p_p\, V_p,
  \end{aligned}
\end{equation}
together with the quasi steady state constraint
\begin{equation}
  \label{eq:quasi-steady-state-flux-equation}
  S^x_y\, V_y + S^x_x\, V_x - \alpha\, \varepsilon\, S^x_p \,V_p = 0.
\end{equation}
In \eqref{eq:reduced-model-fluxes} and \eqref{eq:quasi-steady-state-flux-equation}, the reaction flux vectors $V_y$, $V_x$, and $V_p$ are considered as free time-dependent variables to be used in the dynamic optimization, subject to constraints described later in this section.

In order to develop the optimization problem, we need to further elaborate on the structure of the metabolic-genetic networks.
Most reactions in the network will be catalyzed by an enzyme, which needs to be included in the vector of macromolecular species $P$.
For ease of notation, we will assume that the first $m$ components of $P$ correspond to the network's enzymes, and the remaining components to non-enzymatic macromolecules.
Each enzyme catalyzes a set of one or more reactions.
The set of reactions catalyzed by enzyme $P_i$ is denoted by the set of indices
\begin{equation}
  \label{eq:cat-indices}
  \mathrm{cat}(i) = \{ j \in \mathbb N : P_i \textnormal{ catalyzes } V_j \}.
\end{equation}
We also define a vector $b$ which contains the molecular weights of the macromolecules, such that the scalar product $b\T P$ is equal to the cells' dry weight.

We included the following biophysical constraints in the optimization problem:
\begin{itemize}
\item Enzyme capacity constraints:
  Generally, the reaction fluxes catalyzed by an enzyme $P_i$ are limited by upper and lower bounds of the form
  \begin{equation}
    \label{eq:enzyme-constraint0}
    \sum_{j\in\cat(i)} \left|\frac{V_{j}}{k_{\cat,\pm j}}\right| \leq P_i,
  \end{equation}
  where each $k_{\cat,+j}$ ($k_{\cat,-j}$) is the forward (backward) $k_{\cat}$ value for the reaction $V_j$.
  For irreversible fluxes, i.e., non-negatively constrained $V_{j}$, the constraint~\eqref{eq:enzyme-constraint0} can be written as
  \begin{equation}
    \label{eq:enzyme-constraint1}
    h_{c,i}\T V \leq e_i\T P,
  \end{equation}
  where the $j$-th component of the vector $h_{c,i}$ is 
  \begin{equation*}
    (h_{c,i})_j = \left\{
      \begin{aligned}
        & k_{\cat,+j}^{-1} \textnormal{ if } j \in \cat(i) \\
        & 0 \textnormal{ otherwise},
      \end{aligned}
      \right.
  \end{equation*}
  and $e_i$ is a vector with a 1 in its $i$-th entry and zero elsewhere.
  If some of the fluxes catalyzed by the enzyme $P_i$ are reversible, we need to account for positive and negative signs as well as forward and backward $k_{\cat}$ values independently, and the constraint~\eqref{eq:enzyme-constraint0} needs to be written as
  \begin{equation}
    \label{eq:enzyme-constraint2}
    H_{c,i} V \leq E_i P,
  \end{equation}
  where the matrix $H_{c,i}$ is composed of rows such as the vector $h_{c,i}$ in the irreversible case~\eqref{eq:enzyme-constraint1}, but covering all combinations of positive and negative signs for the components corresponding to the reversible fluxes,
  and all rows of the matrix $E_i$ are equal to $e_i\T$.
  As an example, consider a case where enzyme $P_1$ catalyzes $V_1$ reversibly and $V_2$ irreversibly, such that~\eqref{eq:enzyme-constraint0} can be specified as
  \begin{equation}
    \label{eq:enzyme-constraint-example}
    \left| \frac{V_1}{k_{\cat,\pm 1}} \right| + \frac{V_2}{k_{\cat,+2}} \leq P_1.
  \end{equation}
  Then $H_{c,1}$ matrix is then
  \begin{equation}
    \label{eq:enzyme-constraint3}
    H_{c,1} = \begin{pmatrix}
      -k_{\cat,-1}^{-1} & k_{\cat,+2}^{-1} & 0 & \dots & 0\\
      k_{\cat,+1}^{-1} & k_{\cat,+2}^{-1} & 0 & \dots & 0
    \end{pmatrix},
  \end{equation}
  so that both possible signs of $V_1$ in the constraint~\eqref{eq:enzyme-constraint-example} are accounted for,
  and $E_1 = \begin{pmatrix}
    1 & 0 & \dots & 0 \\
    1 & 0 & \dots & 0
  \end{pmatrix}$.

  At the network level, this translates to the constraints
  \begin{equation}
    \label{eq:enzyme-constraint}
    H_c\, V\, \leq \,H_E\, P,
  \end{equation}
  where the matrices $H_c$ and $H_E$ are defined as vertical concatenations of the matrices $H_{c,1}$ to $H_{c,m}$ and $E_1$ to $E_m$, respectively, where $m$ is the number of enzymes in $P$.
\item Biomass-independent flux bounds, for example positivity of irreversible fluxes:
  \begin{equation}
    \label{eq:flux-bounds}
    V_{\min} \,\leq \,V \,\leq\, V_{\max}\,.
  \end{equation}
\item Positivity of molecular species:
  \begin{equation}
    \label{eq:positive-concentrations}
    P\,\geq\, 0, \qquad Y\, \geq\, 0.
  \end{equation}
\item Biomass composition constraints, related to the cell's solvent capacity:
  In addition to enzymes, the macromolecular species $P$ will in general also include molecules which form the structural support of the cells, for example membrane molecules.
  We assumed that this structural components need to be present in a minimal fraction of the total dry weight $b\T P$ of the cells.
  If $P_s$ is a structural biomass component, such a constraint can be written as
  \begin{equation}
    \label{eq:bm-constraint0}
    \psi_s b\T P \leq P_s,
  \end{equation}
  where $\psi_s$ is the minimal fraction of the total dry weight for $P_s$.
  If there are several such constraints, resulting for example from different structural components, they can be written in matrix form as
  \begin{equation}
    \label{eq:bm-constraint}
    H_B\, P \,\leq\, 0,
  \end{equation}
  where each row of $H_B$ is of the form $\varphi_s b\T - e_s\T$, for different indices $s$.
  With this constraint, the amount of structural cell components will impose an upper bound on the feasible enzyme amount.
  In conjunction with the enzyme capacity constraint~\eqref{eq:enzyme-constraint0}, this also sets an upper constraint on the sum of fluxes as in the molecular crowding extension to FBA \cite{BegVaz2007,VazquezBeg2008}.
\end{itemize}
For the dynamic optimization problem, these constraints are \emph{path constraints}, i.e., they must be satisfied at every time point within the optimization horizon.
 
We combined the reduced model \eqref{eq:reduced-model-fluxes}--\eqref{eq:quasi-steady-state-flux-equation} for a metabolic-genetic network with the biophysical constraints into a dynamic optimization problem as follows.
For notational convenience, we introduce the variable $Z = (Y, P)$.
In addition to the biophysical constraints listed above, some objective functionals required to define a \emph{terminal constraint} in the form of a target set $\mathcal Z$, which was imposed on the network's state only at the terminal time.
We then define the set of all dynamic fluxes $V$ that take the network to the target set within an arbitrary non-negative time $t_f$ as
\begin{equation}
  \label{eq:admissible-fluxes}
  \mathcal V(\mathcal Z,Z_0) = \bigcup_{t_f \geq 0} \{ V \in \mathcal{M}[0,t_f] \mid Z(t_f,V,Z_0) \in \mathcal Z \},
\end{equation}
where $\mathcal M[0,t_f]$ is the set of measurable functions of appropriate dimension over the interval $[0,t_f]$, and $Z(t_f,V,Z_0)$ is the solution of the differential equation~\eqref{eq:reduced-model-fluxes} with flux variables $V(t)$ and initial condition $Z(0) = Z_0$ \cite{MackiStr1982}.

The objective functional $J$ in the optimization is given in general form as an integral over the dynamic variables plus a term for the terminal state:
\begin{equation}
  \label{eq:objective}
  J = \int\limits_{0}^{t_f} \Phi(Y(t), P(t),V(t)) \;dt + \Psi(Z(t_f)).
\end{equation}

The dynamic optimization problem is then constructed as
\begin{equation}
  \label{eq:metabolic-opt-problem}
  \begin{aligned}
    \max_{\mathcal V(\mathcal Z, Z_0)} &\ \int\limits_{0}^{t_f} \Phi(Y(t), P(t),V(t)) \;dt + \Psi(Z(t_f)) \\
    \textnormal{s.t. } & \dot{Y} = - S^y_y\, V_y \\
    & \dot{P} = \varepsilon\, S^p_p\, V_p \\
    & Z(0) = Z_0 \\
    \eqref{eq:quasi-steady-state-flux-equation}\quad& S^x_y\, V_y(t) + S^x_x\, V_x(t) - \alpha\, \varepsilon\, S^x_p\, V_p(t) = 0 \\
    \eqref{eq:enzyme-constraint}\quad& H_c\, V(t) \leq H_E\, P(t) \\
    \eqref{eq:flux-bounds}\quad& v_{\min} \leq V(t) \leq v_{\max} \\
    \eqref{eq:positive-concentrations}\quad& Z(t) \geq 0 \\
    \eqref{eq:bm-constraint}\quad& H_B P(t) \leq 0,
  \end{aligned}
\end{equation}
where equation numbers refer to the derivation of the constraints above.

A variety of approaches are available to numerically solve such an optimization problem \cite{VilasBal2012}.
For this study, we applied collocation methods based on a time discretization of the dynamic variables $Y(t)$, $P(t)$, and $V(t)$ \cite{Stryk1993,RazzaghiNaz1998,Biegler2007}.
The discretized problem was then solved with the Python optimization package cvxopt ({http://abel.ee.ucla.edu/cvxopt/}).
Details are given in~\ref{sec:numerical-solution-collocation}.

\section{Analysis of a minimal metabolic-genetic network and the Monod growth kinetics}
\label{sec:analys-minim-metab}

As a first case study to demonstrate the deFBA approach, we considered the minimal nutrient uptake network introduced in~\eqref{eq:minimal-network}, with the approximated dynamics given by~\eqref{eq:reduced-minimal-model-fluxes}--\eqref{eq:minimal-qss-equation}.
We used the deFBA method to compare the optimal solutions for three different objective functionals.
Biomass maximization is a common objective used in classical flux balance analysis \cite{OrthThi2010}.
We incorporated it into our algorithm in two alternative ways: first, as biomass maximization at the end of the optimization horizon, and second, as discounted maximization of the biomass integrated over time.
The maximization of terminal biomass has also been used in dFBA \cite{MahadevanEdw2002}, whereas the discounted biomass objective has been used as an evolutionary fitness measure in a recent analysis of microbial metabolism \cite{Frank2010}.
As third objective, we considered the minimization of the time required to metabolize the available nutrient completely.
The minimization of the substrate consumption time has also been used previously to predict enzyme concentrations in pathway activation \cite{Klipp2002,Oyarzun2009}.

Formally, the three objective functionals were defined as follows.
\begin{itemize}
\item Maximization of biomass at the end of the considered time interval:
  \begin{equation}
    \label{eq:monod-objective-1}
    J_1 = P(t_f)
  \end{equation}
\item Discounted maximization of the biomass integral:
  \begin{equation}
    \label{eq:monod-objective-2}
    J_2 = \int\limits_{0}^{t_f} P(\tau) \,e^{-\varphi\tau} \;d\tau
  \end{equation}
  with a discount parameter $\varphi \geq 0$.
  A positive discount parameter can in general be used to reduce the effect of the terminal time $t_f$ on the objective function, 
since the objective function value is uniformely bounded for varying terminal times, provided that the discount parameter is larger than the maximal growth rate.
\item Minimization of the time required to metabolize the nutrient completely:
  \begin{equation}
    \label{eq:monod-objective-3}
    J_3 = - \int\limits_{0}^{t_f} d\tau = - t_f
  \end{equation}
  with the terminal constraint
  \begin{equation}
    \label{eq:monod-objective-3-terminal}
    Y(t_f) = 0.
  \end{equation}
\end{itemize}

The biophysical inequality constraints for this example were given by positivity of the molecular species and irreversibility of the two fluxes:
\begin{equation}
  \label{eq:mmg-positivity-constraints}
  \begin{aligned}
    0 &\leq Y &\qquad
    0 &\leq P \\
    0 &\leq V_y &\qquad
    0 &\leq V_p \\
  \end{aligned}
\end{equation}
We also constrained the enzymatic capacity according to \eqref{eq:enzyme-constraint0}, where $P$ is considered as an enzyme catalyzing both the uptake reaction $V_y$ and the biomass reaction $V_p$:
\begin{equation}
  \label{eq:mmg-enzymatic-constraint}
    \frac{V_y}{k_{\cat,y}} + \frac{\varepsilon V_p}{k_{\cat,p}}\, \leq \,P\,,
\end{equation}
where $k_{\cat,y}$ and $k_{\cat,p}$ are the catalytic constants for $V_y$ and $V_p$, respectively.

We performed an analytical study of the optimization problem for the objective functionals $J_2$ and $J_3$, which is described in~\ref{sec:analytical-minimal-mg-net}.
The analysis showed that there is a unique optimal solution which is composed of an initial exponential growth phase until the substrate is completely metabolized, followed by a stationary phase of nil growth. 
This is in agreement with the typical growth kinetics of bacterial cultures \cite{Monod1949}. 
We also found that the time $t_s$ at which the culture switches from exponential to stationary growth is:
\begin{equation}
  \label{eq:monod-switching-time}
  t_s = \left(\frac{\alpha}{k_{\cat,y}} + \frac{1}{k_{\cat,p}}\right)\, \log\left(1 + \frac{Y_0}{\alpha P_0}\right),
\end{equation}
where $Y_0$ and $P_0$ are the initial concentrations of metabolites and biomass, respectively.
The analytical study showed that the minimal network will reach the maximal biomass when the nutrient is depleted.
Regarding the objective $J_1$, this means that any solution where the nutrient is depleted at the terminal time $t_f$ is an optimal solution for the objective $J_1$.
The solution for objective $J_2$, for example, is also optimal for objective $J_1$, suggesting that the optimum for $J_1$ may not be unique.
We confirmed the non uniqueness numerically as described next.

We computed numerical solutions for all three objective functionals, using a discretization approach as described in \ref{sec:numerical-solution-collocation}.
The optimization results are shown in Figure~\ref{fig:result-min-metabogen}.
\begin{figure}
  \centering
 \begin{tabular}{ccc}
   $J_1$ & $J_2$ & $J_3$ \\
   \includegraphics[width=0.32\textwidth]{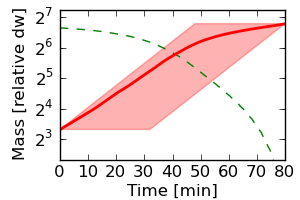} &
   \includegraphics[width=0.32\textwidth]{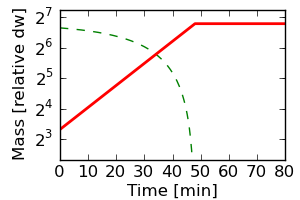} &
   \includegraphics[width=0.32\textwidth]{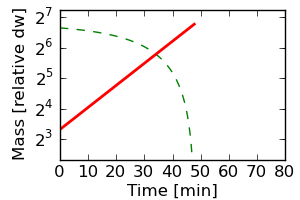} \\
   \includegraphics[width=0.32\textwidth]{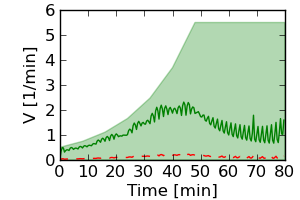} &
   \includegraphics[width=0.32\textwidth]{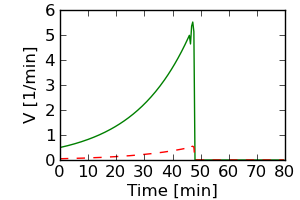} &
   \includegraphics[width=0.32\textwidth]{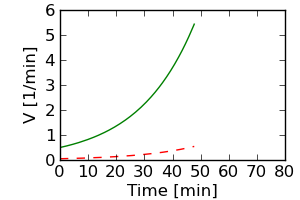} \\
 \end{tabular}
  \caption{Dynamic optimization results for the minimal metabolic-genetic network.
Top row: Optimal time courses of substrate (dashed green line) and biomass (red line) for objective functionals $J_1$, $J_2$, and $J_3$.
Bottom row: Optimal time courses of fluxes ($V_y$ green, $\varepsilon V_p$ dashed red).
Filled regions show the point-wise variability of the optimal solution.
The following parameter values were used: $\alpha=10$, $\varepsilon=0.1$, $k_{\cat,y}=1\ \mathrm{min}^{-1}$, $k_{\cat,p}=1\ \mathrm{min}^{-1}$, $\varphi=0.01\ \mathrm{min}^{-1}$, $t_f=80\ \mathrm{min}$ (for objectives $J_1$ and $J_2$). The initial condition was taken as $Y(0) = 100$ and $P(0) = 1$.
}
  \label{fig:result-min-metabogen}
\end{figure}
In case of objectives $J_2$ and $J_3$, the numerical solution is in agreement with the analytically predicted biphasic growth kinetics.

We also computed the variability of the optimal solutions, using a similar concept as in flux variability analysis \cite{BurgardVai2001}:
We solved an optimization problem as given in~\eqref{eq:metabolic-opt-problem}, but with the original objective functional as an additional constraint, and the minimization or maximization of a flux or state variable at one time point as new objective.
By computing the variability of the optimal solutions point-wise in time, a significant variability of the optimal solutions is observed for the objective $J_1$, whereas the uniqueness of the solutions for objectives $J_2$ and $J_3$ manifests itselfs in the fact that there is no point-wise variability in the numerical solutions (bottom row of Figure~\ref{fig:result-min-metabogen}).

Since the proposed optimization method does not assume any specific reaction kinetics, it yields optimal solutions that, in principle, may not be realizable with a plausible biochemical mechanism. 
To check whether the optimal biphasic growth profile is consistent with a realistic kinetic law, we compared it to the simulated growth curves obtained with a typical kinetic model.
We extended the minimal network to a kinetic model with reaction rates given by the Michaelis-Menten law:
\begin{equation}
  \label{eq:monod-rate-laws}
  \begin{aligned}
    V_y &= \frac{k_{y} \,P\; Y}{K_y\, \vartheta_e + Y}\\
    \varepsilon V_p &= \frac{k_{p}\, P \,X}{K_p + X}. 
  \end{aligned}
\end{equation}
For a faithful comparison with the dynamic optimization results, the parameters $k_{y}$ and $k_{p}$ need to satisfy the enzymatic capacity constraint~\eqref{eq:mmg-enzymatic-constraint}.
We aimed for parameter values ensuring that the enzyme $P$ operates at full capacity, corresponding to the constraint~\eqref{eq:mmg-enzymatic-constraint} being satisfied with equality.
Together with the quasi steady state constraint~\eqref{eq:minimal-qss-equation}, we obtained parameter values given by
\begin{equation}
  \label{eq:mm-parameter-bounds}
  \begin{aligned}
    k_{y} &= \left( \frac{1}{k_{\cat,y}} + \frac{1}{\alpha k_{\cat,p}} \right)^{-1} \\
    k_{p} &= \left( \frac{\alpha}{k_{\cat,y}} + \frac{1}{k_{\cat,p}} \right)^{-1}.
  \end{aligned}
\end{equation}
Simulations of the kinetic model (shown in Figure~\ref{fig:monod-simulation}) were practically identical to the optimal solutions for objectives $J_{2}$ and $J_{3}$, both of which were computed without presupposing any specific reaction kinetics.
\begin{figure}[tbp]
  \centering
 \includegraphics[width=5.5cm]{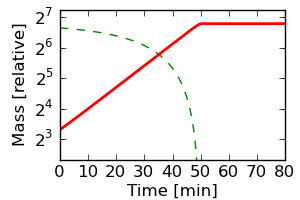}
  \caption{Simulation of the minimal metabolic-genetic network with Michaelis-Menten reaction rates. 
Parameter values were set as $k_{y}=0.5\ \mathrm{min}^{-1}$, $k_{p}=0.05\ \mathrm{min}^{-1}$, $K_y \vartheta_e = K_p = 1$ and initial condition $Y(0) = 100$, $X(0) = 0$, $P(0) = 1$.
}
  \label{fig:monod-simulation}
\end{figure}
These simulation results are similar to the classical Monod model of bacterial growth \cite{Monod1949}.
In \ref{sec:proof-equiv-betw}, we rigorously prove the equivalence of the minimal nutrient uptake network \eqref{eq:minimal-network} and the Monod model. 

\section{Dynamic optimization of a core carbon network}
\label{sec:dynam-optim-core}
\subsection{Network description}
\label{sec:network-description}

We constructed a metabolic-genetic network model (see Figure~\ref{fig:core-network}) as an abstraction of core processes relating carbon uptake and growth. 
It accounts for the uptake of different extracellular species as nutrients, including two carbon sources ${\bio{Carb1}}$ and ${\bio{Carb2}}$, oxygen, fermentation products, and other organic molecules.
The model includes the major anabolic and catabolic processes together with the translational mechanisms for ribosome and enzyme assembly.

The stoichiometry of the exchange and metabolic reactions were taken from the original model proposed in \cite{CovertSch2001}. 
However, in contrast to the work in \cite{CovertSch2001}, we do not include regulatory interactions but instead added catalytic enzymes, structural macromolecules, and ribosomes as biomass components with appropriate biomass production reactions.
We used deFBA to predict the adaptation dynamics by optimizing the reaction rates without the need to include regulatory interactions.

\begin{figure}[tbp]
  \centering
  \includegraphics[width=1\linewidth]{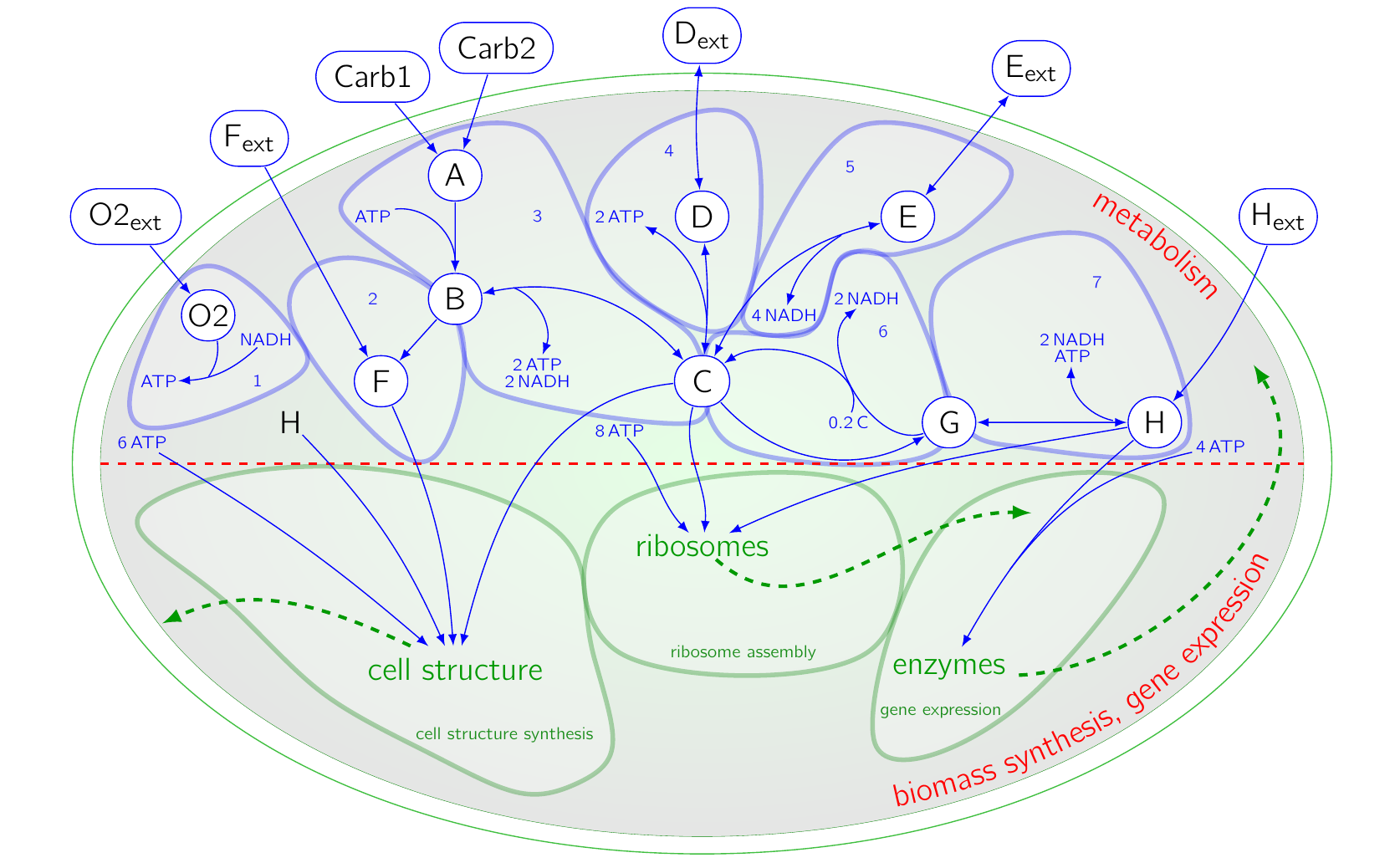}
  \caption{Schematic diagram of the abstract core metabolic-genetic network.
Upper part of the diagram shows metabolic processes, the lower part biomass / gene expression processes.
Abstract representations of core metabolic processes are labelled as follows: (1) Electron transport chain, (2) Lipid metabolism, (3) Glycolysis, (4) Acetate fermentation, (5) Ethanol fermentation, (6) Citric acid cycle, (7) Amino acid metabolism.
Extracellular species correspond to different nutrients and waste products as detailed in the main text.
Intracellular species ${\bio{A}}$--${\bio{H}}$ correspond to metabolite pools,
and ${\bio{O2}}$ is intracellular oxygen.
}
  \label{fig:core-network}
\end{figure}

The model is detailed in Table~\ref{tab:core-network-metabolic}--\ref{tab:core-network-biomass},
together with the utilized enzyme catalytic constants.
To get reasonable flux bounds on reactions describing diffusive exchange across the plasma membrane, we defined the structural component ${\bio{S}}$ as the enzymatic macromolecule for these reactions, together with an appropriate rate constant for diffusion over the plasma membrane or through pore proteins.
In Table~\ref{tab:core-network-metabolic}, enzymes are denoted as ${\bio{T_i}}$ for transport enzymes and ${\bio{E_j}}$ for catalytic enzymes. The macromolecules for biomass reactions and diffusive transport are the ribosome ${\bio{R}}$ and the structural component ${\bio{S}}$, respectively. 
The catalytic constants of the enzymes are based on typical $k_{\cat}$ values in metabolism \cite{Bar-EvenNoo2011,SchomburgCha2002,MiloJor2010}.
While the original model \cite{CovertSch2001} did not specify which reactions were modelled as reversible,
our model includes reversible amino acid biosynthesis, and reversible synthesis and secretion of waste products $\bio{D}$ and $\bio{E}$.

The catalytic constants of the biomass reactions, corresponding to protein translation, are based on the measured translation rate of 17 amino acids/s in \textit{E.\ coli} \cite{YoungBre1976}.
The catalytic constant in the model for each protein translation reaction is the ratio of this translation rate and the number of amino acid residues in the specific protein.
The ${\bio{ATP}}$ requirements for the gene expression reactions are deduced from a flux-based analysis of gene expression \cite{AllenPal2003}.
The stoichiometry of amino acids in the biomass reactions is close to the observed average protein length of 300 amino acids \cite{BrocchieriKar2005}, but with a higher stoichiometry for reactions which are meant to be a condensed description of a complex pathway.

Despite its relative simplicity, the model includes a range of metabolic processes typically found in microbes, such as uptake of alternative carbon sources, aerobic and anaerobic glycolysis, and the uptake and synthesis of lipids and amino acids.
The network includes fermentation products ${\bio{D}}$ and ${\bio{E}}$, which can also be reutilized as nutrients when oxygen is available.
The amino acid ${\bio{H}}$ can be used directly for the synthesis of proteins, or be diverted to the carbon metabolism.
The network is parametrized such that ${\bio{Carb1}}$ should be expected to be the preferred carbon source as compared to ${\bio{Carb2}}$, since the corresponding uptake reaction ${\bio{Carb1}} \rightarrow {\bio{A}}$ is modeled with a higher enzymatic efficiency and a lower enzyme production cost than ${\bio{Carb2}} \rightarrow {\bio{A}}$.
The biosynthetic pathways ${\bio{B}} \rightarrow {\bio{F}}$ and ${\bio{G}} \rightarrow {\bio{H}}$ are modeled to require a substantially larger investment in terms of enzyme production cost compared to the uptake pathways ${\bio{F_{\ext}}} \rightarrow {\bio{F}}$ and ${\bio{H_{ext}}} \rightarrow {\bio{H}}$.
By the ${\bio{C}}$-${\bio{G}}$ cycle illustrated in Figure~\ref{fig:core-network}, the carbon sources can be completely catabolized under aerobic conditions.
Otherwise, they can be catabolized to the fermentation products ${\bio{D}}$ and ${\bio{E}}$.

\begin{table}[tbp]
  \centering
  \caption{Exchange and metabolic reactions, together with rate constants for the core metabolic-genetic network.
  Reversible reactions use equal forward and backward $k_{cat}$ values.}
  \renewcommand{\arraystretch}{1.01}\small
  \begin{tabular}{|l|l|r|}
    \hline
    Reaction & Enzyme & $k_\mathrm{cat} \hspace{0.2em} (\mathrm{min}^{-1})$ \\ \hline \hline
    \multicolumn{3}{|l|}{Exchange reactions} \\ \hline
    ${\bio{Carb1}} \rightarrow {\bio{A}}$ & ${\bio{T_{C1}}}$ & $3000$ \\
    ${\bio{Carb2}} \rightarrow {\bio{A}}$ & ${\bio{T_{C2}}}$ & $2000$ \\
    ${\bio{F_{ext}}} \rightarrow {\bio{F}}$ & ${\bio{T_{F}}}$ & $3000$ \\
    ${\bio{O2_{ext}}} \rightarrow {\bio{O2}}$ & ${\bio{S}}$ & $1000$ \\
    ${\bio{D}} \leftrightarrow {\bio{D_{ext}}}$ & ${\bio{S}}$ & $1000$ \\
    ${\bio{E}} \leftrightarrow {\bio{E_{ext}}}$ & ${\bio{S}}$ & $1000$ \\
    ${\bio{H_{ext}}} \rightarrow {\bio{A}}$ & ${\bio{T_{H}}}$ & $3000$ \\ 
    \hline\hline
    \multicolumn{3}{|l|}{Metabolic reactions} \\ \hline
    ${\bio{A}} + {\bio{ATP}} \rightarrow {\bio{B}}$ & ${\bio{E_{B}}}$ & $1800$ \\
    ${\bio{B}} \rightarrow {\bio{C}} + 2\,{\bio{ATP}} + 2\,{\bio{NADH}}$ & ${\bio{E_{C}}}$ & $1800$ \\
    ${\bio{B}} \rightarrow {\bio{F}}$ & ${\bio{E_{F}}}$ & $1800$ \\
    ${\bio{C}} \rightarrow {\bio{G}}$ & ${\bio{E_{G}}}$ & $1800$ \\
    ${\bio{G}} \rightarrow 0.8 \,{\bio{C}} + 2\, {\bio{NADH}}$ & ${\bio{E_{N}}}$ & $1800$ \\
    ${\bio{C}} \leftrightarrow 2\,{\bio{ATP}} + 3\, {\bio{D}}$ & ${\bio{E_{D}}}$ & $1800$ \\
    ${\bio{C}} + 4\, {\bio{NADH}} \leftrightarrow 3\,{\bio{E}}$ & ${\bio{E_{E}}}$ & $1800$ \\
    ${\bio{G}} + {\bio{ATP}} + 2\, {\bio{NADH}} \leftrightarrow {\bio{H}}$ & ${\bio{E_{H}}}$ & $1800$ \\
    ${\bio{NADH}} + {\bio{O}} \rightarrow {\bio{ATP}}$ & ${\bio{E_{T}}}$ & $1800$ \\
    \hline
  \end{tabular}
  \label{tab:core-network-metabolic}
\end{table}
\begin{table}[tbp]
  \centering
  \caption{Biomass reactions, weight vector $b$ and initial conditions for biomass product, together with the rate constants for the core metabolic-genetic network.
    All biomass reactions are catalyzed by the ribosome \bio{R}.}
  \renewcommand{\arraystretch}{1.01}\small
  \begin{tabular}{|l|r|r|r|}
    \hline
    Reaction & $b$ & $P(0) \, / \, (\mu g l^{-1})$ & $k_\mathrm{cat} \hspace{0.2em} (\mathrm{min}^{-1})$ \\ \hline \hline
    $400\, {\bio{H}} + 1600\, {\bio{ATP}} \rightarrow {\bio{T_{C1}}}$ & $4$ & $17.0$ & $2.5$ \\
    $1500\,  {\bio{H}} + 6000\, {\bio{ATP}} \rightarrow {\bio{T_{C2}}}$ & $15$ & $0.0$ & $0.67$ \\
    $400 \, {\bio{H}} + 1600\, {\bio{ATP}} \rightarrow {\bio{T_{F}}}$ & $4$ & $0.0$ & $2.5$ \\
    $400\, {\bio{H}} + 1600\, {\bio{ATP}} \rightarrow {\bio{T_{H}}}$ & $4$ & $0.0$ & $2.5$ \\
    $500\, {\bio{H}} + 2000\, {\bio{ATP}} \rightarrow {\bio{E_{B}}}$ & $5$ & $28.3$ & $2$ \\
    $500\,  {\bio{H}} + 2000\, {\bio{ATP}} \rightarrow {\bio{E_{C}}}$ & $5$ & $25.8$ & $2$ \\
    $1000\, {\bio{H}} + 4000\, {\bio{ATP}} \rightarrow {\bio{E_{D}}}$ & $10$ & $3.6$ & $1$ \\
    $1000\,  {\bio{H}} + 4000\, {\bio{ATP}} \rightarrow {\bio{E_{E}}}$ & $10$ & $0.0$ & $1$ \\
    $1500\,  {\bio{H}} + 6000\, {\bio{ATP}} \rightarrow {\bio{E_{F}}}$ & $15$ & $2.5$ & $0.67$ \\
    $500\,  {\bio{H}} + 2000\, {\bio{ATP}} \rightarrow {\bio{E_{G}}}$ & $5$ & $29.7$ & $2$ \\
    $2500\,  {\bio{H}} + 10000\, {\bio{ATP}} \rightarrow {\bio{E_{H}}}$ & $25$ & $14.7$ & $0.4$ \\
    $500\,  {\bio{H}} + 2000\, {\bio{ATP}} \rightarrow {\bio{E_{N}}}$ & $5$ & $15.0$ & $2$ \\
    $500\,  {\bio{H}} + 2000\, {\bio{ATP}} \rightarrow {\bio{E_{T}}}$ & $5$ & $52.1$ & $2$ \\
    $2000\,  {\bio{H}} + 4000\, {\bio{C}} + 16000\, {\bio{ATP}} \rightarrow {\bio{R}}$ & $60$ & $29.2$ & $0.2$ \\
    $250\,  {\bio{H}} + 250\, {\bio{C}} + 250\, {\bio{F}} + 1500\, {\bio{ATP}} \rightarrow {\bio{S}}\qquad$ & $7.5$ & $233$ & $3$ \\
        \hline
  \end{tabular}
  \label{tab:core-network-biomass}
\end{table}
\begin{table}[tbp]
  \centering
  \caption{Initial nutrient conditions ($\mathrm{mM}$), oxygen inflow $V_O$ ($\mathrm{mM\,s}^{-1}$), oxygen turnover $\gamma_O$ ($\mathrm{s}^{-1}$), initial biomass $b\T P(0)$ ($\mathrm{g}\,\mathrm{l}^{-1}$), and discount factor $\varphi$ ($\mathrm{min}^{-1}$) for scenarios 1--3}
    \renewcommand{\arraystretch}{1.2}
  \begin{tabular}{|l|ccccccccccc|}
  \hline
    Scenario\qquad & ${\bio{Carb1}}$ & ${\bio{Carb2}}$ & ${\bio{O2_{ext}}}$ & ${\bio{D_{ext}}}$ & ${\bio{E_{ext}}}$ & ${\bio{F_{ext}}}$ & ${\bio{H_{ext}}}$ & $V_O$ & $\gamma_O$ & $b\T P(0)$ & $\varphi$ \\ \hline
    1 & $2$ & $30$ & $50$ & 0 & 0 & 0 & 0 & $20$ & $0.4$ & $0.005$ & $0.1$ \\
    2 & $50$ & 0 & $5$ & 0 & 0 & 0 & 0 & $2$ & $0.4$ & $0.005$ & $0.1$ \\
    3 & $50$ & 0 & $50$ & 0 & 0 & $5$ & $5$ & $20$ & $0.4$ & $0.005$ & $0.3$ \\
    \hline
  \end{tabular}
  \label{tab:scenarios-initial}
\end{table}

\subsection{Dynamic optimization of the core carbon network}
\label{sec:dynam-optim-core-methods}

We used our deFBA algorithm to predict the substrate concentrations, biomass composition, and growth kinetics under three metabolic adaptation scenarios: carbon switch, oxygen limitation and rich medium growth.
Since the analysis of the minimal network suggested that maximizing the discounted biomass integral is a biologically reasonable objective in a dynamic setting, we used this objective functional for the analysis of the core carbon network.
We modeled each scenario by manipulating the initial conditions of the model as detailed in Table \ref{tab:scenarios-initial}.

We computed the initial biomass composition by resource balance analysis (RBA) \cite{GoelzerFro2011} so as to yield maximal aerobic growth rate on ${\bio{Carb1}}$ alone.
The initial biomass composition thus corresponded to a cellular state pre-adapted to such an environment.
RBA determines the biomass distribution for the maximal growth rate in steady state exponential growth,
thus avoiding spurious transients caused by a non-optimal initial composition of biomass.

We assumed the culture to be run as an aerated batch process.
The extracellular oxygen dynamics were modeled by the differential equation
\begin{equation}
  \label{eq:oxygen-turnover}
  \frac{d}{dt}\, {\bio{O2_{ext}}} = V_O - \gamma_O\, {\bio{O2_{ext}}},
\end{equation}
where $V_O$ is the oxygen inflow and $\gamma_O$ the ventilation rate, with different values per Scenario as given in Table~\ref{tab:scenarios-initial}.
Oxygen turnover described by~\eqref{eq:oxygen-turnover} is added to the full biological model for the dynamic optimization.

The objective functional is the discounted biomass integral
\begin{equation}
  \label{eq:core-network-objective}
  J = \int\limits_{0}^{t_f} b\T \,P(t)\, e^{-\varphi t} \;dt,
\end{equation}
where $b$ is the weight vector for the individual biomass components as given in Table~\ref{tab:core-network-biomass}, $P$ contains the reaction products of the biomass reactions listed in Table~\ref{tab:core-network-biomass}, and $\varphi$ is the discount factor as given in Table~\ref{tab:scenarios-initial} for the three scenarios.

In addition to the enzyme capacity constraint, we used a biomass composition constraint to ensure that the structural component $S$ makes up for at least 35 \% of total biomass:
\begin{equation}
  \label{eq:structure-constraint}
  0.35 b\T P \leq \bio{S}.
\end{equation}

\subsection{Scenario 1: carbon switch}
\label{sec:carb-switch-scen}

In this scenario, we studied cellular growth under both carbon sources, with a low concentration for a preferred carbon source ${\bio{Carb1}}$, and a high concentration of the other carbon source ${\bio{Carb2}}$.
The optimization predicted four distinct growth phases, labelled a--d in Figure~\ref{fig:core-network-results-carbon}. 
In phase a, cells grew exclusively on ${\bio{Carb1}}$.
After its depletion, they switched to ${\bio{Carb2}}$ uptake (phase b).
The optimization predicted a nutrient uptake pattern comparable to catabolite repression \cite{Bettenbrock2006}, where a preferred carbon source is completely consumed before cells switch to the non-preferred carbon source.
In the growth phases a and b, cells produced the waste metabolite ${\bio{D}}$.
When both carbon sources were consumed, the optimization predicted re-consumption of the previously excreted waste metabolite ${\bio{D}}$ (phase c), thus being able to sustain growth, though at a significantly lower rate.
The stationary phase d was reached after complete consumption of the substrate.

The predictions seem to indicate a significant intracellular reorganization well before the complete depletion of the second carbon source (see Figure~\ref{fig:core-network-results-carbon}\,B--C).
This means that the optimal response suggests a cellular adaptation to the impending nutrient depletion, reminiscent of changes in gene regulation observed in \textit{S. cerevisiae} just before a glucose--gluconeogenic switch \cite{ZamparKuem2013}.
It is known that some microorganisms, for example \textit{Bacillus subtilis}, monitor the ratio of population size to nutrient availability by quorum sensing \cite{BischofsHug2009}, which could allow them to predict an upcoming nutrient depletion and appropriately adjust their gene expression levels before starvation.

\begin{figure}[tbp]
  \centering
  \includegraphics[width=\linewidth]{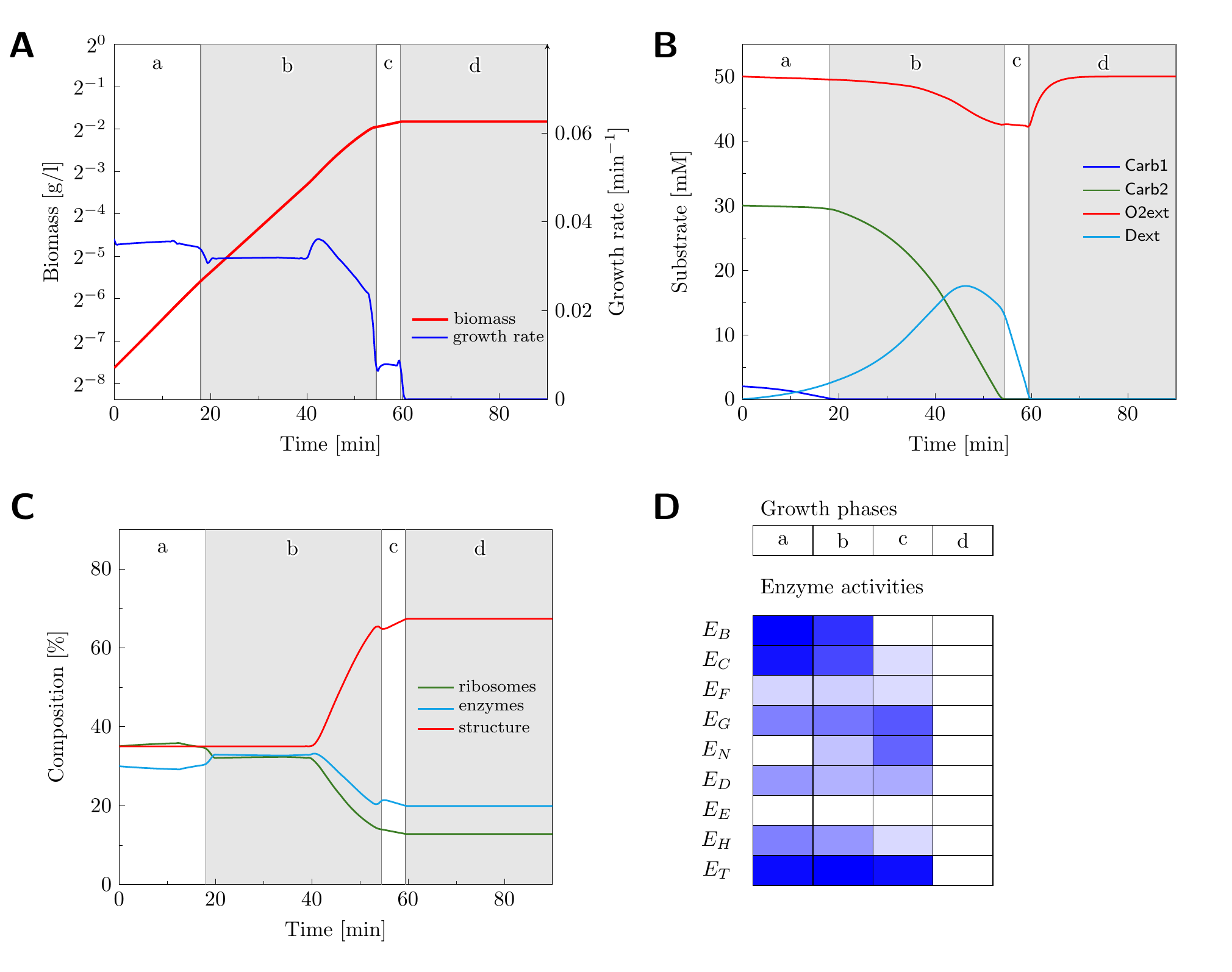}
  \caption{Dynamic optimization results for the core metabolic-genetic network in Scenario 1 (two carbon sources). Time intervals a-d show different growth regimes according to substrate availability. 
A: Biomass and growth rate. 
B: Concentrations of extracellular metabolites. 
C: Dry weight percentages for aggregate cellular components.  
D: Enzyme activity during growth phases a--d, averaged over each phase. Stronger blue denotes higher relative enzyme activity within a phase. See Table~\ref{tab:core-network-metabolic} for enzyme labels.
}
  \label{fig:core-network-results-carbon}
\end{figure}

Figure~\ref{fig:core-network-results-carbon}D illustrates the predicted reorganization of the metabolic network in terms of the enzyme activities.
Glycolysis, represented by the enzyme ${\bio{E_B}}$, was only active in phases a and b.
The ${\bio{C}}$-${\bio{G}}$ cycle, involving enzymes ${\bio{E_G}}$ and ${\bio{E_N}}$, was most active in phase c, where ${\bio{ATP}}$ production from glycolysis ceased and had to be substituted by respiration.
There was a significant drop in activity of the amino acid synthesis pathway represented by ${\bio{E_H}}$ in phase c, presumabely related to the lower enzyme and ribosome biomass fraction in that phase.

\begin{figure}[tbp]
  \centering
  \includegraphics[width=\linewidth]{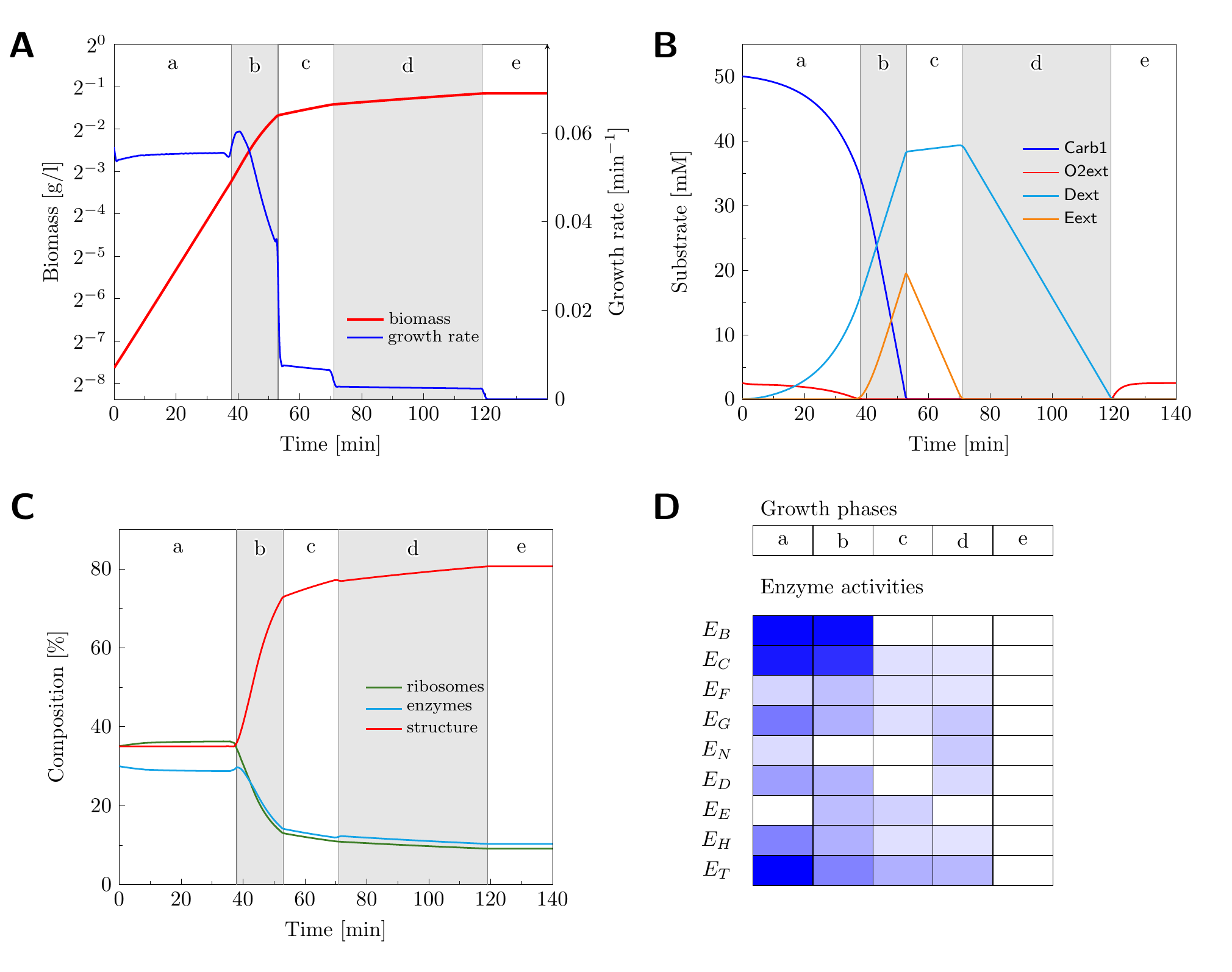}
  \caption{Dynamic optimization results for the core metabolic-genetic network in Scenario 2 (oxygen limitation). Phases a-e show different growth regimes according to substrate availability. 
A: Biomass and growth rate. 
B: Concentrations of extracellular metabolites. 
C: Dry weight percentages for aggregate cellular components. 
D: Enzyme activity during growth phases a--e, averaged over each phase. Stronger blue denotes higher relative enzyme activity within a phase.
 }  
  \label{fig:core-network-results-anaerobic}
\end{figure}

\subsection{Scenario 2: oxygen limitation}
\label{sec:aero-anaero-scen}

In this scenario, we studied cell growth under limited oxygen availability.
The deFBA predictions display five growth phases, labelled a--e in Figure~\ref{fig:core-network-results-anaerobic}.
In phase a, cells grew aerobically on ${\bio{Carb1}}$.
In phase b, oxygen was depleted and cells continued to grow anaerobically, producing both waste metabolites ${\bio{E}}$ and ${\bio{D}}$.
There was also a significant drop in the growth rate during phase b.
The growth phases c and d show that there is also a preferential order in the re-metabolization of waste products.

With the small continuous supply of oxygen that is considered in the model, cells first grew exclusively on ${\bio{E}}$ (phase c), consuming ${\bio{D}}$ only after ${\bio{E}}$ has been depleted (phase d). Phase e was then the stationary phase.

From the enzyme activities shown in Figure~\ref{fig:core-network-results-anaerobic}\,D, we got similar results as in scenario 1 concerning glycolysis and amino acid synthesis.
According to the oxygen limitation, there was a significant drop in the respiratory activity, represented by the enzyme ${\bio{E_T}}$, after phase~a.

\subsection{Scenario 3: transition to and growth in rich medium}
\label{sec:rich-med-scen}

In this scenario, we studied the effect of adding amino acids (${\bio{H_{ext}}}$) and lipids (${\bio{F_{ext}}}$) to the medium for cells grown previously on ${\bio{Carb1}}$ alone.
Initially, cells were assumed to be pre-adapted to a medium without amino acids and lipids, and the optimization predicted an initial transient where cells adapted to the new medium. 
The model predicts expression of amino acid and lipid transporters so as to shift from synthesis to uptake of amino acids and lipids from the medium.
We then observed five growth phases, labelled a--e in Figure~\ref{fig:core-network-results-rich}. 

In phase a, where also the generic amino acid ${\bio{H}}$ is available in the medium, cells grew with a significantly higher rate than in the other scenarios.
This higher growth rate presumably results from decreased enzyme cost for biosynthesis pathways during this phase.
Also, during phases a and b the ribosomal biomass fraction was significantly increased (Figure~\ref{fig:core-network-results-rich}\,C), while the enzymatic biomass fraction was reduced.
This agrees with previous studies about the dependence of the growth rate on global cellular parameters, where increased mRNA and decreased protein fractions have been associated with increased growth rate \cite{KlumppZha2009}.
After depletion of external ${\bio{F}}$ and ${\bio{H}}$, the ratio between ribosomes and enzymes returned to the original state, while the cells continued to grow on ${\bio{Carb1}}$ alone (phase c).
When ${\bio{Carb1}}$ was depleted, cells re-metabolized the waste product ${\bio{D}}$ (phase d), before entering stationary phase (e).

An additional observation from the enzyme activities shown in Figure~\ref{fig:core-network-results-rich}\,D is that glycolysis was most active only in phase b, presumably due to the need of precursor molecules for amino acid synthesis after the extracellular supply of ${\bio{H}}$ had been used up.

\begin{figure}[tbp]
  \centering
  \includegraphics[width=\linewidth]{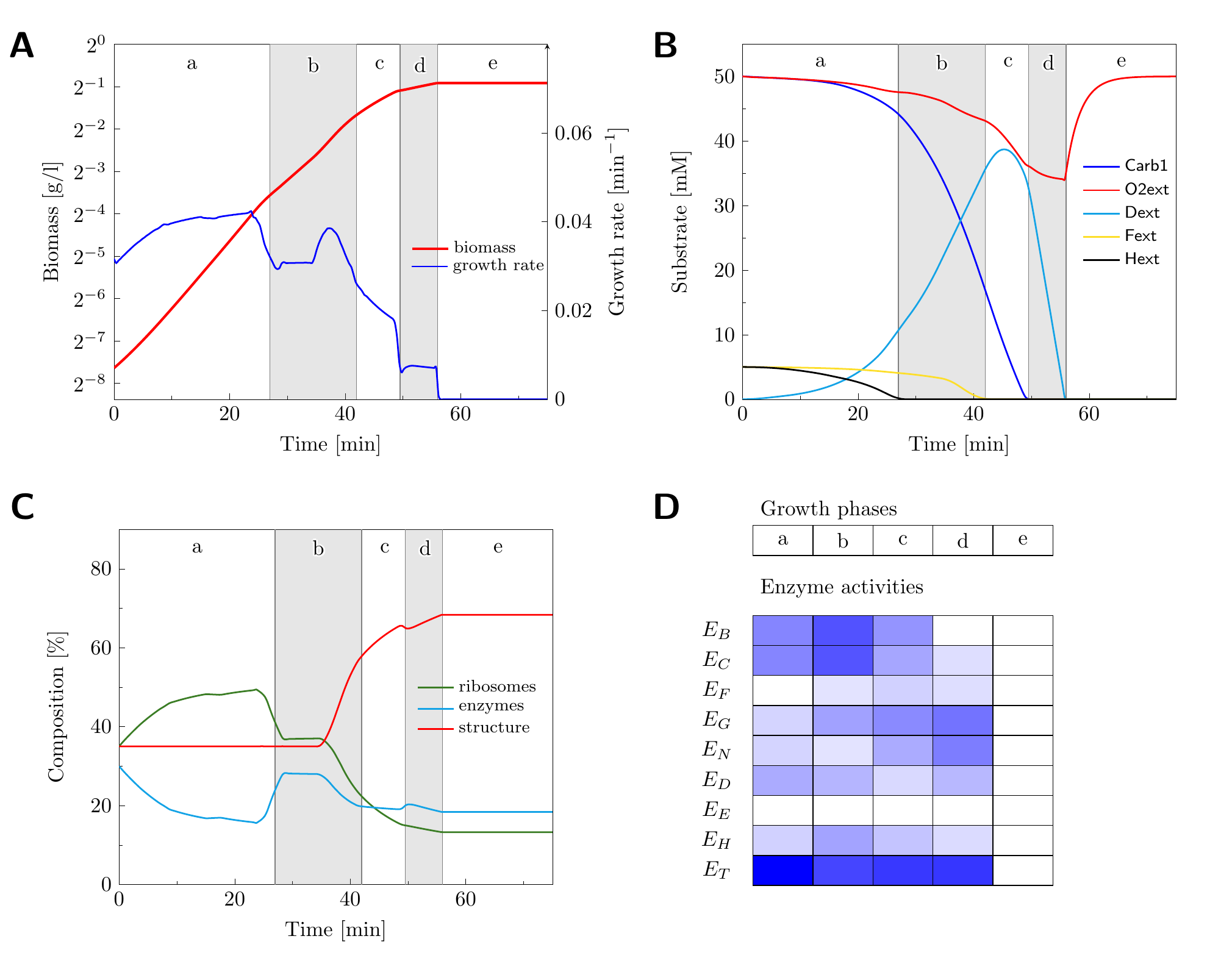}
  \caption{Dynamic optimization results for the core metabolic-genetic network in Scenario 3 (rich medium). Phases a-e show different growth regimes according to substrate availability. 
A: Biomass and growth rate. 
B: Concentrations of extracellular metabolites. 
C: Dry weight percentages for aggregate cellular components.  
D: Enzyme activity during growth phases a--e, averaged over each phase. Stronger blue denotes higher relative enzyme activity within a phase.
}
  \label{fig:core-network-results-rich}
\end{figure}

\section{Discussion}
\label{sec:conclusions}

In this paper we have presented a dynamic optimization approach for metabolic networks coupled with gene expression.
The proposed deFBA method can be used to predict the dynamic adaptation of intracellular metabolites and biomass composition under environmental perturbations.
Based on a biophysically motivated approximation by separation of timescales, we obtained a reduced model for metabolic networks coupled with biomass dynamics. 
In the reduced model, the metabolic fluxes are computed from a quasi steady state approximation, while the nutrient concentrations, enzyme expression levels and other biomass components are described by a system of differential equations.

The inclusion of a detailed biomass composition with constraints on fluxes from enzyme capacity was previously proposed in RBA \cite{GoelzerFro2011} and ME networks \cite{LermanHyd2012} for steady state models.
We developed a dynamic optimization method, called dynamic enzyme-cost flux balance analysis (deFBA), which includes the biomass composition and related enzyme capacity constraints in a dynamic model.
The algorithm allows predicting the time courses of metabolic fluxes and biomass concentrations from an optimality principle.
Our approach provides a generalization to the currently established dynamic flux balance analysis (dFBA) reported in \cite{MahadevanEdw2002}, and offers two advantages.
Firstly, it can readily account for constraints on enzyme levels and their biosynthetic cost. 
Secondly, it is based on computationally efficient linear optimization techniques and thus seems well suited for large scale networks. 
Computational costs can become particularly limiting in dFBA, as it is based on non-linear optimization techniques that are hard to scale with network size.
In addition, deFBA solves one optimization problem for the dynamics over the complete time interval, thus also avoiding numerical problems in differential equations constrained by linear programs \cite{HoffnerHar2013}, as encountered in iterative FBA or the static optimization approach of dFBA \cite{VarmaPal1994,MahadevanEdw2002}.

The deFBA approach allows predicting changes in the enzyme expression levels from an optimization principle alone, without the need of explicit models for gene regulatory interactions as proposed in other FBA-related methods \cite{CovertSch2001,CovertPal2002}.
We exploited this feature by using deFBA to predict metabolic adaptations in two biologically relevant metabolic-genetic systems: a minimal nutrient uptake network and a larger core carbon uptake system.

For the minimal network we showed numerically and analytically that the optimal growth kinetics are composed of an exponential and a stationary phase.
These biphasic growth kinetics appear in the solutions that maximize the discounted biomass integral or minimize the nutrient metabolization time. 
We subsequently found that these optimal growth kinetics are essentially equivalent to the classical Monod model of bacterial growth \cite{Monod1949} and can be accurately modeled with uptake kinetics following a Michaelis-Menten law.
This result can be interpreted as a rigorous derivation of the Monod growth model, which was so far an empirical model based on observed growth dynamics.
The close correspondence of the Monod growth model with the dynamic optimization results suggests that maximization of the discounted biomass integral or the minimization of the metabolization time are biologically plausible objectives in a dynamic context, and in fact both have been previously used in evolutionary \cite{Frank2010} or metabolic pathway studies \cite{Klipp2002,Oyarzun2009}, respectively. 

The plausibility of other objective functionals should also be examined in deFBA, similarly as has been done for static flux balance analysis \cite{SchuetzKue2007}.
However, some objectives which have been proposed there, such as the maximization of ${\bio{ATP}}$ production, seem to be more appropriate for networks focussing on the metabolic level alone, and may not be appropriate for metabolic-genetic networks.
A dynamic optimization approach opens up the possibility of using other objective functions that do not apply in a static setting, such as the minimization of the metabolization time introduced in Section~\ref{sec:analys-minim-metab}.

We also considered the maximization of the final biomass in the minimal network, as this was one of the original objectives proposed with dFBA \cite{MahadevanEdw2002}.
The predictions, however, showed a significant variability due to non-uniqueness of the optimal solution. 
Non-unique optimal solutions are a common problem with all constraint-based models \cite{MahadevanSch2003}. 
In networks with two equivalent parallel pathways, classical FBA would not be able to predict which of the two is active.
This situation can be remedied by considering different enzyme costs for the two pathways, as first suggested in \cite{GoelzerFro2011} for static optimization, and in our study for the dynamic case (Section~\ref{sec:carb-switch-scen}).
The use of static objective functionals in a dynamic setting can aggravate the problem of non-uniqueness.
For example, as we observed in Figure~\ref{fig:result-min-metabogen}, when maximizing the terminal biomass, a dynamic approach may not be able to decide whether a single pathway is active early or late within the considered time horizon.
The alternative objective functionals that we proposed in this study remedied the problem of non-unique optimal solutions in the time domain, as growth at an earlier time would give a better objective functional value than growth at a later time.

Our study on the core carbon network used a larger model that includes the uptake of different extracellular species (including nutrients, oxygen, and organic precursor molecules) together with some of the main energy turnover processes, and the assembly of ribosomes and enzymes.
We applied the deFBA method to predict the growth kinetics and time courses for substrates and biomass composition that maximized the discounted biomass integral in three scenarios.
These scenarios were chosen from classical examples of metabolic adaptation processes: the switch from one carbon source to another, growth under reduced oxygen availability, and the transition of cells to a rich medium, which subsequently gets depleted of nutrients.
The resulting growth kinetics reproduce a number of known biological observations, such as the overall cellular composition under different growth conditions \cite{KlumppZha2009} or a hierarchy of preferences for different carbon sources \cite{ChangSma2004}.
The results show that a significant change in gene regulation before an impending nutrient depletion would be optimal for cellular growth.
In view of recent experimental results on gene regulation during the glycolytic-gluconeogenic switch \cite{ZamparKuem2013}, we suggest growth optimality over the full time period as theoretical explanation for regulatory activity before the switching time point.

The metabolic part of the model in Section~\ref{sec:dynam-optim-core} was first analyzed with regulatory FBA (rFBA) in \cite{CovertSch2001}.
We compared these results to our results obtained by the deFBA method.
Both approaches result in a sequence of growth phases with a distinct metabolic flux pattern.
For the specific scenarios, both the carbon switch and the anaerobic growth have also been studied in \cite{CovertSch2001} and gave similar results, for example the preference for one carbon source over the other.
Importantly, in rFBA, this preference resulted from explicitly building it into the regulatory logic, while with deFBA, it followed implicitly from the different enzymatic efficiencies and metabolic costs of the alternative pathways.
While our results predict the re-metabolization of fermentation products, similar to an observed acetate re-utilization in \textit{E.\ coli} \cite{VarmaPal1994}, which has also been reproduced in iterative FBA \cite{VarmaPal1994} and dynamic FBA \cite{MahadevanEdw2002},
the corresponding reactions seemed to be modelled as irreversible in the rFBA study \cite{CovertSch2001}, which would have prevented the re-metabolization there.

An important general distinction between rFBA and deFBA concerns the biological knowledge required to build the network models.
In rFBA, fluxes are constrained by Boolean rules modeling regulatory mechanisms, which have to be known in the modeling step.
With deFBA, regulatory interactions are not included in the model, but the specific enzymes for each reaction together with their metabolic production costs are added to the network and thus need to be known.
Importantly, these two approaches do not exclude each other: it should be well possible to construct network models with both regulatory constraints and constraints from the capacity and metabolic costs of individual enzymes.

Although we have focussed on the metabolic constraints relating to enzymatic capacity, the deFBA framework readily allows for inclusion of thermodynamic constraints on metabolic fluxes \cite{HenryBro2007}. 
Moreover, a recent study has suggested that constraints on gene regulatory mechanisms may also be relevant and contribute to some mismatch between observed gene expression and fitness levels predicted from the theoretical optimum based on metabolic constraints \cite{PriceDeu2013}. 
The inclusion of suitable constraints on the regulatory mechanisms in a metabolic optimization framework is an open problem, but we suggest that it may become tractable in a dynamic optimization setup such as deFBA, for example by including metabolic costs of gene regulation.
In addition, the objective functions considered here do not take the robustness against unpredictable changes in the environment into account, which may also contribute to slower growth than predicted from a pure growth rate optimization \cite{GoelzerFro2011a}.

In essence, the approach presented here can predict the temporal regulation of gene expression from an optimization principle, without requiring any knowledge of regulatory interactions.
It yields predictions for biomass dynamics in metabolic adaptations, while respecting constraints of enzymatic capacity and mass conservation.



\appendix

\section{Derivation of the long time scale models}
\label{sec:derivation-long-time-scale}

Since the reaction fluxes $V_i(t,y,x,P)$ and the cellular volume $\vartheta_c(t,P)$ are assumed to be slowly varying,
we could rewrite them on the long time scale as
\begin{equation}
  \label{eq:activity-volume-long-timescale}
  \begin{aligned}
    V_i\left(\varepsilon^{-1} T, y, x, P\right) &= \tilde V_i\left(T, y, x, P\right) \\
    \vartheta_c\left(\varepsilon^{-1} T, P\right) &= \tilde \vartheta_c(T, P) \\
  \end{aligned}
\end{equation}
even in the limit $\varepsilon \rightarrow 0$.

On the long time scale, the metabolic-genetic network is then rewritten as
\begin{equation}
  \label{eq:net-long-timescale}
  \begin{aligned}
    y^\prime &= - \frac{1}{\varepsilon\, \vartheta_e}\, S^y_y\, \tilde V_y\left(T,y, \frac{X}{\tilde\vartheta_c(T,P)},P\right)\\
    \varepsilon X^\prime &= S^x_y \, \tilde V_y\left(T, y, \frac{X}{\tilde\vartheta_c(T,P)}, P\right) + S^x_x \, \tilde V_x\left(T, \frac{X}{\tilde\vartheta_c(T,P)}, P\right) \\&\quad- \alpha\, \varepsilon\, S^x_p \, \tilde V_p\left(T, \frac{X}{\tilde\vartheta_c(T,P)}, P\right) \\
    P^\prime &= S^p_p \, \tilde V_p\left(T, \frac{X}{\tilde\vartheta_c(T,P)}, P\right),
  \end{aligned}
\end{equation}
where we used $X^\prime = dX/dT$ to denote the time derivative on the long time scale.

Based on the model~\eqref{eq:net-long-timescale}, the quasi steady state equation for Condition~1 in Section~\ref{sec:rigor-deriv-quasi} is given by
  \begin{equation}
    \label{eq:quasi-steady-state-equation-app}
       \begin{aligned}
   S^x_y \, \tilde V_y\left(T, y, \frac{X}{\tilde\vartheta_c(T,P)}, P\right) + S^x_x \, \tilde V_x\left(T, \frac{X}{\tilde\vartheta_c(T,P)}, P\right)&  \\
   -\alpha\, \varepsilon\, S^x_p \, \tilde V_p\left(T,\frac{X}{\tilde\vartheta_c(T,P)}, P\right) &= 0,
    \end{aligned}
  \end{equation}
and the boundary layer model~\eqref{eq:boundary-layer} is more explicitly written as
  \begin{equation}
    \label{eq:boundary-layer-app}
       \begin{aligned}
    \dot X =& S^x_y \, \tilde V_y\left(T, y, \frac{X}{\tilde\vartheta_c(T,P)}, P\right) + S^x_x \, \tilde V_x\left(T, \frac{X}{\tilde\vartheta_c(T,P)}, P\right) \\&- \alpha\, \varepsilon\, S^x_p \, \tilde V_p\left(T, \frac{X}{\tilde\vartheta_c(T,P)}, P\right),
       \end{aligned}
  \end{equation}
where the slow variables $T$, $y$, and $P$ are considered as constant in~\eqref{eq:quasi-steady-state-equation-app} and~\eqref{eq:boundary-layer-app}.

The reduced model is constructed as
\begin{equation}
  \label{eq:reduced-model-long-timescale}
  \begin{aligned}
    y^\prime &= - \frac{1}{\varepsilon\, \vartheta_e}\, S^y_y\, \tilde V_y\left(T, y, \frac{q(T,y,P)}{\tilde\vartheta_c(T,P)}, P\right)\\
    P^\prime &= S^p_p \, \tilde V_p\left(T, \frac{q(T,y,P)}{\tilde\vartheta_c(T,P)}, P\right),
  \end{aligned}
\end{equation}
where $X = q(T,y,P)$ is the solution to the quasi steady state constraint~\eqref{eq:quasi-steady-state-equation-app}.
Going back to the original $t$ time scale and to units of molar amount for $Y$, the reduced dynamics are
\begin{equation}
  \label{eq:reduced-model}
  \begin{aligned}
    \dot Y &= - S^y_y\, V_y\left(t, \frac{Y}{\vartheta_e}, \frac{q(\varepsilon t,Y/\vartheta_e,P)}{\vartheta_c(t,P)}, P\right)\\
    \dot P &= \varepsilon S^p_p \, V_p\left(t, \frac{q(\varepsilon t,Y/\vartheta_e,P)}{\vartheta_c(t,P)}, P\right).
  \end{aligned}
\end{equation}

\section{Numerical solution of the dynamic optimization problem by collocation}
\label{sec:numerical-solution-collocation}

In this section, we discuss the numerical solution of the optimization problem~\eqref{eq:metabolic-opt-problem} in more detail.
As a first step, let us consider the case where the terminal time $t_f$ is fixed \emph{a priori}.
In the collocation scheme, the time interval $[0,t_f]$ is divided into $N$ equally sized intervals, each of length
\begin{equation}
\label{app-eq:collocation-intervals}
h = \frac{t_f}{N}\,.
\end{equation}
Within each interval, $K$ collocation points are determined. 
All collocation points are given by the sequence
\begin{equation}
  \label{app-eq:collocation-points}
  t_{1,1}, t_{1,2}, \dotsc, t_{1,K}, t_{2,1}, \dotsc, t_{N,K}.
\end{equation}
Within each interval, the $q$-th collocation point is at position $r_q$ (relative to the interval $[-1,1]$), where $r_q$ is determined by the collocation scheme and order.
In the computational experiments for this study, we used Radau collocation points of order 2 and 3, determined by zeros of the Legendre polynomials \cite{Biegler2007}.
The collocation points are thus computed as
\begin{equation}
  \label{app-eq:collocation-points-formula}
  t_{i,q} = (i-1)\, h + (r_q + 1)\, \frac{h}{2}.
\end{equation}

The flux variable and the derivative of the species variable are discretized by the following interpolation scheme:
\begin{equation}
  \label{app-eq:variable-discretization}
  \begin{aligned}
    V(t) &= \sum_{q=1}^K v_{i,q}\, L_q\left(\frac{2 t - 2 t_{i-1} - h}{h}\right), \qquad t_{i-1} \leq t \leq t_i \\
    \dot Z(t) &= \sum_{q=1}^K \dot z_{i,q}\, L_q\left(\frac{2 t - 2 t_{i-1} - h}{h}\right), \qquad t_{i-1} \leq t \leq t_i,
  \end{aligned}
\end{equation}
where $L_q$, $q=1,\ldots,K$ are suitable interpolation functions defined on the interval $(-1,1)$.
In this study, we used the Lagrange polynomials
\begin{equation}
  \label{app-eq:lagrange}
  L_q(r) = \prod_{1\leq i \leq K,\ i\neq q}^K \frac{r - r_i}{r_q - r_i}
\end{equation}
as interpolation functions \cite{GargPat2009}.
The boundaries of the time intervals are given by $t_i = i h$, $i=1,\dotsc,N$ and $t_0 = 0$.

The species variable $Z$ is discretized at the boundaries of the $N$ intervals in time, and its value within an interval is approximated by integrating over the time derivative:
\begin{equation}
  \label{app-eq:state-approx}
  \begin{aligned}
    Z(t_{i-1} + \tau) &= z_{i-1} + \int\limits_{0}^{\tau} \dot Z(t_{i-1} + s) \;ds \\[0.5\baselineskip]
    &= z_{i-1} + \sum_{q=1}^K \dot z_{i,q} \int\limits_0^\tau L_q(2s/h - 1) \;ds \\[0.5\baselineskip]
    &= z_{i-1} + h/2 \sum_{q=1}^K \dot z_{i,q} \int\limits_{-1}^{r(\tau)} L_q(s)\; ds,
  \end{aligned}
\end{equation}
with $r(\tau) = 2\tau/h - 1$.

The continuous optimization problem~\eqref{eq:metabolic-opt-problem} is now approximated by a finite-dimensional problem, in which the optimization is done over the vector $w \in \Real^{NK(n+m)+Nn}$, defined by
\begin{equation}
  \label{app-eq:discrete-optimization-vector}
  w = (v_{1,1}, v_{1,2}, \dotsc, v_{N,K}, \dot z_{1,1}, \dot z_{1,2} \dotsc, \dot z_{N,K}, z_1, \dotsc, z_N),
\end{equation}
with $v_{i,q}$, $\dot z_{i,q}$, and $z_i$, $i=1,\dotsc,N$, $q=1,\dotsc,K$ corresponding to the interpolation coefficients in \eqref{app-eq:variable-discretization} and \eqref{app-eq:state-approx}.

The discretized optimization problem is written as
\begin{equation}
  \label{app-eq:discrete-oc-problem-fixedtime}
  \begin{aligned}
    \max_{w} &\ c\T w + e \\[0.5\baselineskip]
    \textnormal{s.t. } & M_e w = c_e\\[0.5\baselineskip]
    & M_i w \leq c_i,
  \end{aligned}
\end{equation}
where the vector $c$ and number $e$ stem from the discretization of the the objective functional, the equality constraint matrix $M_e$ and vector $c_e$ from collocation of the differential equation and the initial condition, and the inequality constraint matrix $M_i$ and vector $c_i$ from the path and terminal constraints in \eqref{eq:metabolic-opt-problem}.
The optimization problem \eqref{app-eq:discrete-oc-problem-fixedtime} is a linear program and can directly be solved by common optimization software.

\section{Analytical results for the minimal metabolic-genetic network}
\label{sec:analytical-minimal-mg-net}
\subsection{Problem statement}
\label{sec:problem-statement}

In this section we show how to obtain the analytic solution of the optimization problems in Section~\ref{sec:analys-minim-metab}. 
We seek to solve the following dynamic optimization problem:
\begin{align}
	\max_{V_{y}} J_{i},\qquad i=2,3.
\end{align}
in the time interval $\lc0,t_{f}\rc$ subject to the model dynamics
\begin{align}
	\dot{Y} &= -V_{y},\\
	\dot{P} &= \varepsilon V_{p}
\end{align}
and the quasi steady state constraint
\begin{equation}
  V_y = \alpha \varepsilon V_p,
\end{equation}
together with the positivity constraints
\begin{align}
	Y\geq 0,\quad P\geq0,\quad V_{y}\geq0,\quad V_{p}\geq0
\end{align}
and the enzyme capacity constraint
\begin{align}
	\frac{V_{y}}{k_{\cat,y}} + \frac{\varepsilon V_{p}}{k_{\cat,p}} \leq P.
\end{align}
The objective functions are the discounted biomass integral
\begin{align}
	J_{2} &= \int_{0}^{t_{f}} P(\tau)e^{-\varphi\tau}\mathrm{d}\,\tau,
\end{align}
and the time needed to metabolize all the nutrient
\begin{align}
	J_{3} &= -t_{f}.
\end{align}
In the case of objective $J_{2}$, the final time $t_{f}$ is pre-specified. In the case of $J_{3}$, the final time $t_{f}$ is free and subject to an additional terminal constraint $Y(t_{f})=0$.

Using the quasi steady state approximation $V_{y}=\alpha\varepsilon V_{p}$, the  model dynamics can be written as 
\begin{align}
	\dot{Y} &= -V_{y},\label{eq:ydot}\\
	\dot{P} &= V_{y}\slash\alpha,\label{eq:pdot}
\end{align}
and we rewrite the enzyme capacity constraint as 
\begin{align}
	V_{y} &\leq P\slash K,
\end{align}
with 
\begin{equation}
  \label{eq:K}
  K=1\slash k_{\cat,y} + 1\slash (\alpha k_{\cat,p}). 
\end{equation}
In the next sections we show how to obtain analytic solutions to both of these problems.

\subsection{Solution for objective $J_{2}$}

 To avoid ambiguities, from now on the star $^{*}$ denotes optimal trajectories. We first note that $V_{y}\geq 0$ implies that $P\geq0$, because $P_{0}\geq0$ and $\dot{P}=V_{y}\slash\alpha \geq 0$ for all $V_{y}\geq0$ and $\alpha>0$. In addition, by the quasi steady state constraint $V_{y}=\alpha\varepsilon V_{p}$, we have that $V_{y}\geq 0$ implies that $V_{p}\geq 0$, and therefore the optimization problem can be simplified to
\begin{align}
	\max J_{2}
\end{align}
subject to the dynamics \eqref{eq:ydot}--\eqref{eq:pdot}, and the constraints
\begin{align}
	V_{y}\leq P\slash K,\, V_{y}\geq 0,\,\,\mathrm{and}\,\,Y\geq0.
\end{align}
The optimal nutrient dynamics are $\dot{Y}^{*}=-V_{y}^{*}$, which can be solved by integration
\begin{align}
	Y^{*}(t) &= Y_{0} - \int_{0}^{t} V_{y}^{*}(\tau)\mathrm{d}\tau.
\end{align}
This means that for any initial condition $Y_{0}>0$, the optimal nutrient concentration will reach $Y^{*}=0$ only if there exists $0<t_{s}\leq t_{f}$ such that
\begin{align}
	\int_{0}^{t_{s}} V_{y}^{*}(\tau)\mathrm{d}\tau &= Y_{0}.\label{eq:condts}
\end{align}

We will now obtain the optimal uptake rate assuming that equation \eqref{eq:condts} does not have a solution for $t_{s}$. We will then use this solution \emph{a posteriori} to analyze the case when \eqref{eq:condts} does have a solution.
If equation \eqref{eq:condts} does not have a solution, the nutrient does not deplete and therefore constraint $Y\geq0$ never becomes active in the optimization interval $\lc0,t_{f}\rc$. We can thus ignore this constraint and reduce the optimization problem to
\begin{align}
	\max J_{2}
\end{align}
subject to \eqref{eq:ydot}--\eqref{eq:pdot} and $0\leq V_{y}\leq P\slash K$. Since $J_{2}$ grows as $P(t)$ grows (pointwise in time) and $\dot{P}=V_{y}\slash\alpha\geq 0$, the optimal uptake rate $V_{y}^{*}$ must be maximal pointwise in time while respecting $0\leq V_{y}\leq P\slash K$. Intuitively, this means that the optimal rate satisfies $V_{y}^{*}=P^{*}\slash K$, but we can also alternatively use a proof by contradiction as follows. Assume that in the optimal solution the constraint $V_{y}\leq P\slash K$ is not active, so that $V_{y}^{*}=P^{*}\slash K - \delta(t)$ with $\delta(t)>0$ on a subinterval of $\lc 0,t_{f}\rc$. Substituting this optimal rate in $\dot{P}^{*}=V_{y}^{*}\slash\alpha$ we get
\begin{align}
	\dot{P}^{*} &= \frac{P^{*}}{K\alpha} - \frac{\delta}{\alpha}.\label{eq:proof1}
\end{align}
Equation \eqref{eq:proof1} is a linear inhomogeneous differential equation with solution
\begin{align}
	P^{*}(t) &= P_{0}e^{\frac{t}{K\alpha}} - e^{\frac{t}{K\alpha}}\int_{0}^{t} e^{\frac{-\tau}{K\alpha}}\delta(\tau)\mathrm{d}\tau.
\end{align}
Using the definition of $J_{2}=\int_{0}^{t_{f}}P(t)e^{-\varphi t}\mathrm{d}t$, we can compute the value of the corresponding optimal objective as
\begin{align}
	J_{2}^{*} &= P_{0}\int_{0}^{t_{f}} e^{\frac{t}{K\alpha}}e^{-\varphi t}\dt - \underbrace{\int_{0}^{t_{f}}e^{\frac{t}{K\alpha}}e^{-\varphi t}\lp \int_{0}^{t} e^{\frac{-\tau}{K\alpha}}\delta(\tau)\mathrm{d}\tau\rp\dt}_{>0},\\
	&< P_{0}\int_{0}^{t_{f}} e^{\frac{t}{K\alpha}}e^{-\varphi t}\dt,\label{eq:proof2}
\end{align}
which contradicts the optimality of $J_{2}^{*}$, and therefore we conclude that $\delta(t)=0$ on the interval $\lc 0,t_{f}\rc$, apart from a set of zero measure. The optimal uptake rate and biomass concentration are then
\begin{align}
	V_{y}^{*}(t) &= \frac{P_{0}}{K}e^{\frac{t}{K\alpha}},\,\,\mathrm{for}\,\, t\in\lc 0,t_{f}\rc,\label{eq:exponential1}\\
	P^{*}(t) &= P_{0}e^{\frac{t}{K\alpha}},\,\,\mathrm{for}\,\, t\in\lc 0,t_{f}\rc.\label{eq:exponential2}
\end{align}
From these expressions we can \emph{a posteriori} obtain a condition for equation \eqref{eq:condts} to have a solution for $t_{s}$. Substituting the optimal rate $V_{y}^{*}$ in \eqref{eq:condts}, we get an equation for $t_{s}$:
\begin{align}
	\alpha\lp e^{\frac{t_{s}}{K\alpha}} - 1\rp &= \frac{Y_{0}}{P_{0}},
\end{align}
and solving for $t_{s}$ we get
\begin{align}
	t_{s} &= K\alpha \log\lp1 + \frac{Y_{0}}{\alpha P_{0}}\rp.
\end{align}
Therefore, if $t_{f}< t_{s}$ the nutrient never depletes and the optimal solution are the exponentials in \eqref{eq:exponential1}--\eqref{eq:exponential2}. Conversely, if $t_{f}\geq t_{s}$ the equation in \eqref{eq:condts} has a solution for $t_{s}$ and we can establish that
\begin{align}
	Y^{*}(t) &=0,\,\,\mathrm{for}\,\, t\in\lc t_{s},t_{f}\rc,\label{eq:condtsY}\\
	V_{y}^{*}(t) &=0,\,\,\mathrm{for}\,\, t\in\lc t_{s},t_{f}\rc.\label{eq:condtsVy}
\end{align}
Note that \eqref{eq:condtsY} is true because when $Y^{*}$ reaches zero at $t=t_{s}$, it can only be made positive by a negative $V_{y}^{*}$, which would violate the positivity constraint $V_{y}\geq0$. Similarly, if \eqref{eq:condtsVy} is not true, then $V_{y}^{*}$ must become positive for some non empty time interval after $t=t_{s}$, but this would imply that $Y^{*}<0$ in that interval, thereby violating the positivity constraint $Y\geq0$.  We thus conclude that for $t\geq t_{s}$, the network enters stationary phase and the biomass remains constant:
\begin{align}
	P^{*}(t) &=P^{*}(t_{s}),\,\,\mathrm{for}\,\, t\in\lc t_{s},t_{f}\rc.
\end{align}

\subsection{Solution for objective $J_{3}$}
In this case the optimal solution for objective $J_{3}$ can be obtained with similar arguments as the one for $J_{2}$. Maximization of $J_{3}$ is equivalent to minimization of the time it takes the nutrient to deplete ($t_{f}$). Since $\dot{Y}=-V_{y}<0$ it follows that $t_{f}$ decreases as $V_{y}$ grows pointwise in time. This essentially means that the constraint $V_{y}\leq P\slash K$ must be active for $t\in\lc 0,t_{f}\rc$ and therefore the optimal solution is $V_{y}^{*}=P^{*}\slash K$. Following a similar procedure as in the case of $J_{2}$,  we have that the optimal uptake rate and biomass concentrations are exponentials
\begin{align}
	V_{y}^{*}(t) &= \frac{P_{0}}{K}e^{\frac{t}{K\alpha}},\,\,\mathrm{for}\,\, t\in\lc 0,t_{f}\rc,\label{eq:exponential3}\\
	P^{*}(t) &= P_{0}e^{\frac{t}{K\alpha}},\,\,\mathrm{for}\,\, t\in\lc 0,t_{f}\rc.\label{eq:exponential4}
\end{align}
Note that, analogously to equation \eqref{eq:condts}, in this case the time needed for nutrient depletion, $t_{f}$, satisfies
\begin{align}
	\int_{0}^{t_{f}} V_{y}^{*}(\tau)\mathrm{d}\tau &= Y_{0}.\label{eq:condtf}
\end{align}
Substituting \eqref{eq:exponential3} in \eqref{eq:condtf} and solving for the optimal $t_{f}$ we get
\begin{align}
	t_{f} &= K\alpha \log\lp1 + \frac{Y_{0}}{\alpha P_{0}}\rp.
\end{align}

\section{Equivalence between the minimal metabolic-genetic network and the Monod growth kinetics}
\label{sec:proof-equiv-betw}

Since the nutrient uptake and metabolization is done by enzymes, it appears reasonable to assume a common enzymatic rate law such as the Michaelis-Menten rate here.
With this assumption, we obtained the following result.
\begin{proposition}
\label{prop:monod-model}
  If the uptake reaction $V_y$ in the minimal metabolic-genetic network \eqref{eq:minimal-network} is given by the Michaelis-Menten rate law with $P$ as an enzyme and $Y$ as a substrate,
  \begin{equation}
    \label{eq:minimal-uptake}
    V_y = \frac{k_{\cat} P y}{K_m + y},
  \end{equation}
  then the long timescale approximation~\eqref{eq:reduced-minimal-model-fluxes} of the minimal metabolic-genetic network is equivalent to the Monod growth kinetics given by
\begin{equation}
  \label{eq:monod-model}
  \begin{aligned}
    \dot Y &= -\frac{1}{\varrho}\mu(y) P \\
    \dot P &= \mu(y) P.
  \end{aligned}
\end{equation}
\end{proposition}
In the model~\eqref{eq:monod-model}, $Y$ denotes the substrate, $P$ is the biomass, $y = Y / \vartheta_e$ the substrate concentration,
\begin{equation}
  \label{eq:monod-empirical-growth}
  \mu(y) = \frac{\mu_{\max} \,y}{K_y + y}
\end{equation}
is the empirical growth rate, and $\varrho \geq 0$ is the yield coefficient \cite{DunnHei2003}.

Thus the Monod growth model is theoretically explained by a Michaelis-Menten type nutrient uptake reaction together with optimality of cell growth with respect to an objective of time-minimal growth or a discounted biomass integral.

\begin{proof}[Proof of Proposition~\ref{prop:monod-model}]
  From the quasi-steady state condition
\begin{equation}
  \label{app-eq:minimal-net-qss}
  V_y = \alpha \varepsilon V_p,
\end{equation}
we have
  \begin{equation}
    \label{eq:monod-qss}
    V_p = \frac{1}{\varepsilon\alpha}\frac{k_{\cat} y}{K_m + y}.
  \end{equation}
  Then the long timescale approximation of the minimal metabolic-genetic network is given by
  \begin{equation}
    \begin{aligned}
      \dot Y &= - \frac{k_{\cat} P y}{K_m + y} \\
      \dot P &= \frac{1}{\alpha} \frac{k_{\cat} P y}{K_m + y}.
    \end{aligned}
  \end{equation}
  By comparing this equation to~\eqref{eq:monod-model}, we see that the dynamics are identical when identifying the parameters with
  \begin{equation}
    \label{eq:monod-parameters}
    \begin{aligned}
      \mu_{\max} &= \frac{k_{\cat}}{\alpha} \\
      K_y &= K_m \\
      \varrho &= \frac{1}{\alpha}.
    \end{aligned}
  \end{equation}
\end{proof}

\section*{Acknowledgements}
\label{sec:acknowledgements}

We thank Alexandra Reimers for helpful comments on a previous version of the paper.
Diego Oyarz{\'u}n acknowledges support from an Imperial College London Junior Research Fellowship.

\newpage
\section*{References}


\end{document}